\documentclass[aps,pra,showpacs,twoside,twocolumn,longbibliography,10pt]{revtex4-1}
\usepackage[colorlinks=true, citecolor=blue, urlcolor=magenta ]{hyperref}
\usepackage{epsfig,newlfont,amssymb,amsfonts,amsmath,bm,subfigure,palatino,mathtools,amsthm,braket,times,soul,enumitem,color}
\usepackage[normalem]{ulem}
\newcommand{\stkout}[1]{\ifmmode\text{\sout{\ensuremath{#1}}}\else\sout{#1}\fi}
\usepackage[english]{babel}
\usepackage[utf8]{inputenc}
\usepackage{array}
\usepackage{xcolor}
\usepackage{graphics}
\newtheorem{theorem}{Theorem}
\newtheorem{corollary}{Corollary}

\newcommand{\ketbra}[2]{|#1\rangle \langle #2|}
\def\Tr{\text{Tr}}

\begin{document}


\title{Isocoherent  Work Extraction from Quantum Batteries: Basis-Dependent Response}

\author{Shuva Mondal, Debarupa Saha, and Ujjwal Sen}
\affiliation{Harish-Chandra Research Institute,  A CI of Homi Bhabha National
Institute, Chhatnag Road, Jhunsi, Prayagraj - 211019, India}

\begin{abstract} 
We identify 
a 
connection between quantum coherence and the maximum extractable work from a quantum battery, and to this end, we define the \textit{coherence-constrained maximal work} (CCMW) as the highest  amount of work extractable via coherence-preserving unitaries, optimized over all quantum states with fixed coherence in a given dimension. For qubit systems, we derive an analytical relation between the CCMW and the input coherence, defined with respect to an arbitrary fixed basis. Strikingly, we find that for fixed quantum coherence in the energy eigenbasis, the maximal extractable work decreases with increase of coherence. In contrast, when quantum coherence is with respect to a basis for which the Hamiltonian possesses off-diagonal elements, and has equal diagonal elements, the CCMW increases with the level of quantum coherence. We numerically observe that the basis-dependent response of the CCMW also persists in higher-dimensional quantum systems. Moreover, we show that even in higher dimensions one can derive closed-form relations between the CCMW and the input quantum coherence within certain numerically-assessed conclusions. We also comment on 
the structure of passive states in an isocoherent scenario, that is, states from which no energy can be extracted under coherence-preserving unitaries.
\end{abstract}
\maketitle

\section{Introduction}



Quantum batteries~\cite{PhysRevE.87.042123,PhysRevLett.111.240401,Binder2015,Giorgi2015,PhysRevLett.118.150601,Friis2018,PhysRevLett.120.117702,PhysRevB.98.205423,PhysRevA.97.022106,Santos2019,PhysRevE.99.052106,PhysRevB.99.035421,PhysRevLett.122.047702,PhysRevB.99.205437,PhysRevLett.122.210601,PhysRevB.100.115142,Alicki2019,Ghosh_2020,Chen_2020,Caravelli_2020,Garc_a_Pintos_2020,Gherardini_2020,Kamin_2020,Juli_Farr__2020,PhysRevA.104.L030402,Zhang_2023,group_saha2023,RevModPhys.96.031001,group_shukla2025,kzvn-dj7v,6kwv-z6fx} refer to quantum systems that can store or deliver energy, governed by the laws of quantum mechanics. Just like in the classical scenario, there are two crucial aspects of quantum batteries: charging~\cite{PhysRevLett.118.150601,PhysRevLett.122.210601,Binder2015,PhysRevE.111.044118,wen2025} and discharging~\cite{PhysRevE.103.042118,PhysRevA.104.042209,PhysRevLett.134.220402}. Besides their theoretical significance, several studies such as~\cite{PhysRevA.106.032428,PhysRevA.106.042601,Rubin_2023,10025593,YANG2024102300} have also explored experimental implementations of quantum batteries.

In this paper, we mainly focus on the discharging facet of a quantum battery. Typically, discharging of a quantum battery is equivalent to energy reduction or extraction from the battery.
Whenever the discharging process is governed by a unitary operation, the change in the energy of the battery corresponds to the total amount of work performed by the battery. The maximum amount of such work that can be extracted by considering all possible unitary operations is termed \emph{ergotropy}~\cite{Allahverdyan_2004}. States of a quantum battery from which a finite amount of ergotropy can be extracted are called \emph{active states}, while those with zero ergotropy are referred to as \emph{passive states}~\cite{Lenard1978,Pusz1978,PhysRevE.90.012127,PhysRevE.91.052133,PhysRevE.92.042147,Brown2016,Sparaciari2017,PhysRevE.98.042121,PhysRevLett.123.190601}.

The performance of a quantum battery is therefore fundamentally determined by the amount of work it can deliver. Several studies have attempted to enhance this performance by exploiting quantum features intrinsic to the battery, 
including entanglement~\cite{Entr1,pro1,Entr2,GGM3} and quantum coherence~\cite{CoAberg,co1}. For instance, Refs.~\cite{PhysRevE.87.042123,PhysRevLett.118.150601,PhysRevB.100.115142,Juli_Farr__2020,PhysRevA.104.L030402,chaki1,chaki2} explored the role of entanglement in the energy extraction process, while Refs.~\cite{PhysRevLett.125.180603,PhysRevLett.129.130602,PhysRevLett.131.030402,PhysRevE.109.044119} demonstrated how the presence of quantum coherence can enhance the performance of a quantum battery. However, an exact quantitative relationship between the amount of these quantum traits and the extractable energy remains largely unexplored. Uncovering this connection could significantly deepen our understanding of quantum thermodynamics and provide valuable insights for the design of efficient quantum batteries, especially under realistic experimental constraints, where the initial quantum features are limited and costly and therefore,  the extraction protocols are required to preserve them.

In this article, we address this gap by establishing a concrete connection between a limited quantum feature intrinsic to the battery state and the amount of extractable energy, under the consideration that the extraction process preserves this feature. The feature we consider is quantum coherence, 
a fundamental trait of quantum systems, often providing advantages in performing various tasks over situations where the trait is absent. 

Specifically, we consider all quantum states of a given dimension having a fixed amount of quantum coherence in a certain basis. We then identify all unitaries that preserve quantum coherence in that basis, and find the maximal work that can be extracted by such unitaries from all states with the fixed quantum coherence. 
%
The maximization is carried out over all input states, pure or mixed, that possess the same level of quantum coherence, and over all unitary operations that preserve this coherence. 
It is important to note that the set of coherence-preserving unitaries may be state-dependent.
We refer to this coherence-constrained maximal work as CCMW, and we explore how it depends on the initial coherence of the quantum battery. Focusing first on qubit batteries, we derive a closed-form relation between the CCMW and the amount of quantum coherence specified in a fixed, but otherwise arbitrary, basis. This relation reveals that the CCMW is always achieved by pure battery states, and it further uncovers a basis-dependent behavior in the work extraction process. Specifically, when the fixed coherence is with respect to the eigenbasis of the Hamiltonian, the CCMW decreases as the initial coherence increases. In contrast, when the Hamiltonian has zero or equal diagonal elements and nonzero off-diagonal elements in the coherence basis, the CCMW increases with increasing coherence.

We demonstrate that this basis-dependent trend in battery performance persists in higher-dimensional  batteries as well. For such systems, we compute the CCMW numerically and analyze its variation with input coherence. Our numerical results show that even in higher-dimensional systems, with a fixed quantum coherence in the energy eigenbasis, it is enough to consider pure input states with fixed coherence to reach the CCMW. Building on these observations, we derive exact expressions for the CCMW–coherence relationship in qutrit batteries, and we propose a general framework for evaluating the CCMW in even larger dimensions. Furthermore, and still remaining with quantum coherence in the energy eigenbasis, we identify passive states with fixed coherence that yield zero extractable energy when the coherence must be conserved. For situations where the Hamiltonian is non-diagonal in the chosen coherence basis, informed by our qutrit-case numerical data, we again derive a closed-form relation linking CCMW to the fixed input coherence. 

The rest of the paper is organized as follows. In Sec.~\ref{sec: prelims}, we provide a brief discussion on the preliminary concepts used in this paper. Specifically, we briefly discuss the $l_1$- norm of quantum coherence in Sec~\ref{sec: l1_norm}, followed by a discussion on ergotropy and passive states in Sec.~\ref{sec: ergotropy&passive_state}. In Sec.~\ref{sec: ccmw}, we present our set-up of energy extraction and formally define CCMW. In Sec.~\ref{sec: qubit_case}, we present the qubit case and analyze how CCMW relates to that of the input quantum coherence. In Sec.~\ref{sec: qudit_system}, we extend our analysis to  higher dimensional systems, focusing on the scenario when the fixed input coherence is with respect to the energy eigenbasis, in Sec.~\ref{subsection: diagonal hamiltonian}. In Sec.~\ref{sec: passive_state}, we present the form of the passive states in our setting. The non-diagonal cases are presented in Sec.~\ref{sec: off_diagonal_hamiltonian} and Sec.~\ref{sec: non-diagonal_&_diagonal}.  And lastly we conclude in Sec.~\ref{sec: conclusion}.



\section{Preliminaries}\label{sec: prelims}
In this section, we review the preliminary concepts employed in our study, including the \(l_1\)-norm of coherence, ergotropy, and passive states. The \(l_1\)-norm of coherence is discussed in detail in the subsection below.

\subsection{$l_{1}$-norm of quantum coherence}\label{sec: l1_norm}
Quantum coherence~\cite{co1,CoAberg} is an intrinsic feature of quantum systems that separates them out from their classical counterparts. In general, quantum coherence refers to the ability of the quantum system to exist in the superposition of multiple basis states simultaneously, with well defined phase relationships between them. In recent year quantum coherence prove to be a useful and crucial resource in various quantum information processing tasks~\cite{RevModPhys.89.041003,wang2024quantumcoherencefundamentalresource}. This makes quantum coherence an intriguing topic to study. Especially how coherence governs the working of quantum technologies is a fascinating area of research. 

There are several measures to quantify quantum coherence~\cite{RevModPhys.89.041003,quantifying_coherence}. In this study, however, we employ the \(l_1\)-norm of coherence, which is defined as the sum of the absolute values of all off-diagonal elements of a quantum state expressed in a given basis.

Let a quantum state \(\rho\) defined on a \(d\)-dimensional Hilbert space \(\mathcal{H}_d\) be expressed in an orthonormal basis \(\{\ket{i}\}\), where \(i=0,1,\dots,d-1\), as
\[
\rho = \sum_{i,j} \rho_{ij}\,\ketbra{i}{j}.
\]
Then, the \(l_1\)-norm of coherence of the state \(\rho\) in that basis is defined as
\[
\text{C}\left(\rho\right) = \sum_{i \ne j} \sqrt{\rho_{ij} \rho_{ij}^{*}} = \sum_{i \ne j} |\rho_{ij}|.
\]
From now on, we will use the term “coherence” to refer to the \(l_1\)-norm of coherence, for convenience. It is important to note that coherence is a basis-dependent quantity, and the amount of quantum coherence in a system varies depending on the basis in which the state is expressed.

Having discussed quantum coherence, we now turn to the two main aspects of a quantum battery, namely the maximum extractable energy and the concept of passive states.

\subsection{Ergotropy and passive states}\label{sec: ergotropy&passive_state} 
There are two crucial concepts that govern the performance of a quantum system as a potential quantum battery~\cite{RevModPhys.89.041003}. These concepts are \emph{ergotropy}~\cite{Allahverdyan_2004,PhysRevE.87.042123} and \emph{passive states}~\cite{Lenard1978,Pusz1978}. In this section, we present a brief discussion of both.

Let \(\rho_{\text{in}}\) denote the initial state of a \(d\)-dimensional quantum system, and let \(H\) be its Hamiltonian; both \(\rho_{\text{in}}\) and \(H\) act on the Hilbert space \(\mathcal{H}_d\). Note here ``in" in the suffix of \(\rho_{\text{in}}\) is used to denote initial. The energy of the state \(\rho_{\text{in}}\) is given by \(\Tr\left[\rho_{\text{in}} H\right]\), where \(\Tr\) denotes the trace of a matrix. Now let $\Lambda$ denote a quantum operation acting on a $d$-dimensional Hilbert space $\mathcal{H}_d$. This operation maps an initial state $\rho_{\text{in}}$ to a final state $\rho_{\text{f}} = \Lambda(\rho_{\text{in}})$, defined in $\mathcal{H}_d$, where the subscript ``f'' denotes ``final''. In such a scenario the energy extractable from $\rho_{\text{in}}$ under the action of $\Lambda$ is given by
\[
E = \operatorname{Tr}\left[ H \left( \rho_{\text{in}} - \Lambda(\rho_{\text{in}}) \right) \right].
\]

If the operation $\Lambda$ is restricted to be a unitary operation $U$, the maximum extractable energy over all such unitary operations is defined as the \emph{ergotropy}, given by
\begin{equation} \label{eq:ergotropy}
\xi(\rho_{\text{in}}) = \max_{U \in \mathcal{U}} \left( \operatorname{Tr}[\rho_{\text{in}} H] - \operatorname{Tr}[U \rho_{\text{in}} U^{\dagger} H] \right),
\end{equation}
where $\mathcal{U}$ denotes the set of unitary operations on $\mathcal{H}_d$. Note that ergotropy is a convex function of the initial state. This can be shown as follows. 

Suppose the initial state can be written as a convex combination of states $\{\rho_j\}$
\[
\rho_{\text{in}} = \sum_j q_j\,\rho_j, 
\quad \text{with} \quad \sum_j q_j = 1,\quad q_j \ge 0.
\]

Using the linearity of the trace operation and the fact that the maximum of a sum is less than or equal to the sum of maxima, we have
\[
\xi( \sum_j q_j \rho_j ) \leq \sum_j q_j \xi(\rho_j).
\]
Now since any state $\rho_{\text{in}}$ can be expressed as a convex mixture of pure states $\rho_j = \ket{\psi_j}\bra{\psi_j}$, it follows that
\[
\xi(\rho_{\text{in}}) \leq \sum_j q_j \xi(\rho_j) \leq \max_j \xi(\rho_j).
\]
The above relation marks the convexity property of ergotropy.

Next, a quantum state $\rho_{\mathrm{in}}$ is called \emph{passive} if no ergotropy can be extracted from it.  Equivalently, for all unitary operations $U$,
\[
\operatorname{Tr}[\rho_{\mathrm{in}} H] \le \operatorname{Tr}[U\,\rho_{\mathrm{in}}\,U^{\dagger} H].
\]

When all unitary operations are allowed to extract energy using the Hamiltonian
\[
H = \sum_i \epsilon_i\,\ketbra{i}{i},
\]
any passive state must satisfy:

1. It is diagonal in the energy eigenbasis.
2. Its populations are non-increasing with increasing energy.

Consequently, the general form of a passive state is
\[
\sigma_{\mathrm{p}} = \sum_i s_i\,\ketbra{i}{i},
\]
where
\[
\sum_i s_i = 1,
\quad
s_i \ge s_j \quad \text{whenever} \quad \epsilon_i \le \epsilon_j.
\]

Having discussed the concept of maximum energy extraction from quantum batteries, we now proceed to our setup, wherein we investigate how the maximum extractable work relates to the coherence of the initial state. To this end, we define the maximum amount of energy extractable from a quantum system of a given dimension, maximized over all states with fixed coherence in a particular basis and all unitary operations that conserve coherence—as the coherence-constrained maximal work (CCMW). We examine how CCMW correlates with coherence. The concept of CCMW is detailed in the section below.

\section{Coherence-constrained maximal work}\label{sec: ccmw}

To explore how coherence governs the maximum amount of energy extractable from a quantum system, potentially used as a quantum battery, we consider the following setup. We fix the initial coherence of the system in a given basis and impose the constraint that the extraction process must preserve this coherence. We restrict ourselves to unitary operations.

Within this setting, we define the coherence constrained maximal work as the highest energy extractable maximized over all states with a fixed coherence level in a given basis (hereafter referred to as the coherence basis) and over all unitary operations that preserve this coherence level. It is important to note that the set of coherence-preserving unitaries may depend on the initial state.
Thus, the CCMW is given by
\begin{equation}\label{eq: our_ergotropy}
    \xi_{d}\left(\mathcal{C}\right):=
\max_{\substack{\rho_{\text{in}} \in \chi^{\mathcal{C}}_d, \\ U_{\mathcal{C}} \in \mathcal{U}_{d}^{\mathcal{C}}\left(\rho_{\text{in}}\right)}}
\left(\Tr\left[\rho_{\text{in}} H\right]-\Tr\left[U_{\mathcal{C}}\rho_{\text{in}}U_{\mathcal{C}}^{\dagger} H\right]\right).
\end{equation}
For a particular choice of basis,
here $\rho_{\text{in}}$ is an initial state with a given coherence $\mathcal{C}$ and the set of all such states is denoted by $\chi^{\mathcal{C}}_{d}$, where the subscript $d$ indicates the dimension of the quantum battery. The $U_{\mathcal{C}}$ denotes the unitary matrices which conserves the coherence of the initial state $\rho_{\text{in}}$, and set of all these unitaries is denoted as $\mathcal{U}^{\mathcal{C}}_d\left(\rho_{\text{in}}\right)$. We denote the unitaries which is used to transform the initial state during the energy extraction process and also preserve the coherence of the state on which they act, as $U_{\mathcal{C}}$. 

The restriction that the coherence remains unchanged during the allowed operations is important, because only then can we unambiguously assign a coherence label to the maximum extractable energy. Consequently, the quantity defined in~\eqref{eq: our_ergotropy} does not depend on the initial state, but only on the fixed amount of coherence.

The CCMW can be written in terms of the final state 
\[
\rho_{\mathrm{f}} = U_{\mathcal{C}}\,\rho_{\mathrm{in}}\,U_{\mathcal{C}}^{\dagger},
\]
which has the same coherence and eigenvalues as \(\rho_{\mathrm{in}}\), as
\begin{equation}\label{eq:our_ergotropy_final_state_form}
\xi_{d}(\mathcal{C})
= \max_{\{\rho_{\mathrm{in}},\rho_{\mathrm{f}}\}\in\chi^{\mathcal{C}}_{d}}
\left(
\Tr[\rho_{\mathrm{in}} H]
- \Tr[\rho_{\mathrm{f}} H]
\right),
\end{equation}
with the constraint that both \(\rho_{\mathrm{in}}\) and \(\rho_{\mathrm{f}}\) are unitarily connected and therefore have the same eigenvalues.

Having defined the CCMW, we introduce another quantity of interest:

\begin{align}
\label{eq:pure_ergotropy}
\xi^p_{d}(\mathcal{C})
&:= \max_{\substack{\rho_{\mathrm{in}} \in \zeta^{\mathcal{C}}_{d}, \\ U_{\mathcal{C}} \in \mathcal{U}^{\mathcal{C}}_d\left(\rho_{\text{in}}\right)}}
\left(
\Tr[\rho_{\mathrm{in}} H] - \Tr[U_{\mathcal{C}} \rho_{\mathrm{in}} U_{\mathcal{C}}^{\dagger} H]
\right) \\
&= \max_{\{\rho_{\mathrm{in}}, \rho_{\mathrm{f}}\} \in \zeta^{\mathcal{C}}_{d}}
\left(
\Tr[\rho_{\mathrm{in}} H] - \Tr[\rho_{\mathrm{f}} H]
\right),
\end{align}
where $\zeta^{\mathcal{C}}_{d}$ is the set of all pure states having the same coherence level $\mathcal{C}$. And superscript $p$ denotes ``pure".

Note that the only difference between $\xi_{d}(\mathcal{C})$ and $\xi^p_{d}(\mathcal{C})$ is that the latter is optimized only over pure states with fixed coherence. In general, one has
\[
\xi_{d}(\mathcal{C}) \ge \xi_{d}^p(\mathcal{C}).
\]

In the subsequent section, we show that for certain dimensions $d$ and for certain choices of the coherence basis, equality can be achieved as
\[
\xi^p_{d}(\mathcal{C}) = \xi_{d}(\mathcal{C}).
\]
In other cases, $\xi^p_{d}(\mathcal{C})$ merely provides a lower bound on the CCMW. Hence, the behavior of the CCMW can be inferred from the properties of $\xi^p_{d}(\mathcal{C})$.

With our setup and key quantities defined, we now proceed to investigate how both CCMW and $\xi^p_{d}(\mathcal{C})$ depend on the fixed initial coherence across different dimensions. We begin with the qubit case, the detailed analysis of which follows below.

\section{Qubit batteries}\label{sec: qubit_case}
In this section, we examine qubit  batteries and present a theorem that provides a closed-form relationship between the unitarily-extracted CCMW and the coherence fixed in an arbitrary basis. The theorem is presented below.

\begin{theorem}\label{theorem: 1}
Consider a qubit  battery, with $H$ being the Hamiltonian of the battery defined on the Hilbert space $\mathcal{H}_2$. Let the coherence $\mathcal{C}$ of this quantum battery be fixed in a given, but otherwise arbitrary, basis also referred to as the coherence basis $\{\ket{0}, \ket{1}\}$. Then, the unitarily extracted CCMW is given as
\begin{equation}
\label{eq: xi2}
    \xi_{2}\left(\mathcal{C} \right)=|h_{1}-h_{3}|\sqrt{1-\mathcal{C}^{2}}+2h_{2}\mathcal{C}, 
\end{equation}
where $h_{1}$ and $h_{3}$ are diagonal entries and $h_{2}$ is absolute value of off-diagonal entry of the Hamiltonian $H$, in the coherence basis.
\end{theorem}

\begin{proof}
Let $\chi^{\mathcal{C}}_{2}$ be the set of all states representing a qubit battery that possess the same amount of coherence $\mathcal{C}$ in a given, but otherwise arbitrary, basis $\{\ket{0},\ket{1}\}$. The subscript $2$ in $\chi^{\mathcal{C}}_{2}$ indicates the dimension of the quantum battery.
The density operator $\rho_{\mathcal{C}}^{a,\theta_1} \in \chi^{\mathcal{C}}_{2}$, corresponding to a quantum state with fixed coherence $\mathcal{C}$, is given by
\begin{align*}
\rho_{\mathcal{C}}^{a,\theta_1} =& \frac{1+a}{2}\ketbra{0}{0} + \frac{1-a}{2}\ketbra{1}{1}\\ &+\frac{\mathcal{C}}{2} e^{-i\theta_{1}} \ketbra{0}{1} + \frac{\mathcal{C}}{2} e^{i\theta_{1}} \ketbra{1}{0},
\end{align*}
where each individual state in the set $\chi^{\mathcal{C}}_{2}$ has a similar form as given above but is distinguished by different values of the parameters $a$ and $\theta_1$,
with $a \in [-\sqrt{1 - \mathcal{C}^2}, \sqrt{1 - \mathcal{C}^2}]$ and $\theta_1 \in [0, 2\pi)$.
 The Hamiltonian of the quantum battery, when written in the basis $\{\ket{0},\ket{1}\}$ is given as
 \begin{align}\label{eq: qubit_hamiltonian}
    H= h_1 \ketbra{0}{0}+h_3 \ketbra{1}{1} +h_{2}e^{-i \theta} \ketbra{0}{1}+h_{2}e^{i \theta}\ketbra{1}{0}.
\end{align} 
 To extract energy from such an initial state $\rho_{\mathcal{C}}^{a,\theta_1}$, keeping its coherence fixed in the basis $\{\ket{0},\ket{1}\}$ we perform coherence-conserving unitary operations $U_{\mathcal{C}}$.
The set of all unitary operators that conserve the coherence of the qubit battery state $\rho_{\text{in}}$ is denoted as $\mathcal{U}_{2}^{\mathcal{C}}\left(\rho_{\text{in}}\right)$, and each element of the set can be written in the form 
\begin{equation*}
U_\mathcal{C}=\sigma^p_x\text{exp}\left[I\beta\sigma_z\right]\sigma^q_x.
\end{equation*} 
Here, $I=\sqrt{-1}$. Note that the expression of $U_{\mathcal{C}}$ is independent of the initial state parameters. In other words, irrespective of the nature of the initial state considered, $U_{\mathcal{C}}$ always preserves the coherence of the initial state in the coherence basis. The operators $\sigma_x$ and $\sigma_z$ appearing in the expression of $U_{\mathcal{C}}$ are given as $\sigma_x = \ketbra{0}{1} + \ketbra{1}{0}$, $\sigma_z = \ketbra{0}{0} - \ketbra{1}{1}$, and $p, q = 0,1$. The parameter $\beta$ can take values in the range $[0, 2\pi)$. Regardless of the parameter values of $U_{\mathcal{C}}$, it can transform $\rho_{\mathcal{C}}^{a,\theta_1}$ only to states of the form
\begin{align*}
U_{\mathcal{C}}\rho_{\mathcal{C}}^{a,\theta_1}U_{\mathcal{C}}^\dagger=\rho_{\mathcal{C}}^{b,\theta_2}=& \frac{1+b}{2} \ketbra{0}{0}+\frac{1-b}{2} \ketbra{1}{1}\\
    & +\frac{\mathcal{C}}{2} e^{-i\theta_{2}} \ketbra{0}{1}+\frac{\mathcal{C}}{2} e^{-i\theta_{2}} \ketbra{1}{0},
\end{align*}
with $b=\pm a$. So, the maximization over all possible $U_{\mathcal{C}} \in \mathcal{U}^{\mathcal{C}}_{2}$ is equivalent to maximization over the final states $\rho_{\mathcal{C}}^{b,\theta_2} \in \chi^{\mathcal{C}}_{2}$. In other words 

\begin{align*}
    &\quad\max_{\rho_{\mathcal{C}}^{a,\theta_1}\in\chi^{\mathcal{C}}_{2},U_{\mathcal{C}}\in\mathcal{U}^{\mathcal{C}}_{2}}\Tr\Big[H(\rho_{\mathcal{C}}^{a,\theta_1}-U_{\mathcal{C}}\rho_{\mathcal{C}}^{a,\theta_1}U^{\dagger}_{\mathcal{C}})\Big]\hspace{-2cm}\\
&=\max_{\{\rho_{\mathcal{C}}^{a,\theta_1},\rho_{\mathcal{C}}^{b,\theta_2}\}\in\chi^{\mathcal{C}}_{2}}\Tr\Big[H(\rho_{\mathcal{C}}^{a,\theta_1}-\rho_{\mathcal{C}}^{b,\theta_2})\Big].
\end{align*}

Thus, the unitarily extracted CCMW in this setup is

\begin{align*}
\xi_{2}\left(\mathcal{C}\right) &=\max_{\{\rho_{\mathcal{C}}^{a,\theta_1},\rho_{\mathcal{C}}^{b,\theta_2}\}\in\chi^{\mathcal{C}}_{2}}\Tr[H(\rho_{\mathcal{C}}^{a,\theta_1}-\rho_{\mathcal{C}}^{b,\theta_2})]\\
&=\max_{a, b, \theta_1, \theta_2} \bigg[\left(h_{1} - h_{3}\right)\frac{a+b}{2}
+ h_{2}\mathcal{C} \cos\left(\theta - \theta_1\right) \\
& + h_{2}\mathcal{C} \cos\left(\theta - \theta_2\right) 
\bigg].
\end{align*}
The second line in the above equation is written with the understanding that, in the present case, maximization over the initial and final states is equivalent to performing the maximization over the state parameters $a$, $b$, $\theta_1$, and $\theta_2$. Note that the three terms in the expression for $\xi_{2}$ are disjoint (i.e., each term depends on different independent parameters). Hence, each term can be maximized independently. This yields
\begin{align*}
\xi_{2}\left(\mathcal{C}\right) = &\max_{a}\left((h_{1} - h_{3})a\right) + h_{2}\mathcal{C} \max_{\theta_1} \left(\cos(\theta - \theta_1)\right)\\
& + h_{2}\mathcal{C} \max_{\theta_2} \left(\cos(\theta - \theta_2)\right).
\end{align*}

Note that, to maximize the first term with respect to $b$, we set $b = a$. To further optimize this term with respect to $a$, one needs to consider the sign of $h_1 - h_3$. For instance, if $h_1 < h_3$, one should take the extreme value $a = -\sqrt{1 - \mathcal{C}^{2}}$, and if $h_1 > h_3$, one should take the maximum value $a = \sqrt{1 - \mathcal{C}^{2}}$. The last two terms are maximized when $\theta = \theta_1$ and $\theta = \theta_2$, respectively. Thus, by setting $a$, $\theta_1$, and $\theta_2$ accordingly, we obtain the general expression of the unitarily extracted CCMW as
\begin{equation*}
\xi_{2}\left(\mathcal{C}\right) = |h_{1} - h_{3}|\sqrt{1 - \mathcal{C}^{2}} + 2h_{2}\mathcal{C}.
\end{equation*}

This completes the proof of Theorem 1.
\end{proof}

It is evident from the above theorem that the optimal initial state, which yields the unitarily extracted CCMW with $a = \pm \sqrt{1 - \mathcal{C}^{2}}$, satisfies the conditions $\Tr[\rho^2] = \Tr[\rho] = 1$. Hence, this state must be a pure state. Note that, in the absence of any constraints, the maximum amount of extractable work always satisfies a convexity relation, as discussed in Sec.~\ref{sec: prelims}.

However, such convexity is not guaranteed in the presence of constraints. Yet, our findings suggest that even under a fixed coherence constraint, the optimal amount of work can still be extracted using pure initial states alone. There is no need to consider mixed quantum states. We summarize this insight in the form of the following corollary.

\begin{corollary}
To unitarily extract CCMW, the optimal initial state corresponding to a qubit battery with coherence $\mathcal{C} \in [0, 1/2]$ fixed in an arbitrary basis $\{\ket{0}, \ket{1}\}$ and having a Hamiltonian $H$ of the nature given in Eq.~\eqref{eq: qubit_hamiltonian}, must be a pure state of the form
\begin{equation*}
\ket{\psi^{i,\theta}_{\mathcal{C}}} = \sqrt{\frac{1 + \delta \sqrt{1 - \mathcal{C}^2}}{2}} \ket{0} + e^{-i \theta} \sqrt{\frac{1 - \delta \sqrt{1 - \mathcal{C}^2}}{2}} \ket{1},
\end{equation*}
where the superscript $i$ denotes the "initial" state, and $\theta$ can vary in $[0,2\pi).$

The optimal final battery state is given by
\begin{equation*}
\ket{\psi^{f,\theta}_{\mathcal{C}}} = \sqrt{\frac{1 - \delta \sqrt{1 - \mathcal{C}^2}}{2}} \ket{0} + e^{-i \theta} \sqrt{\frac{1 + \delta \sqrt{1 - \mathcal{C}^2}}{2}} \ket{1},
\end{equation*}
where the superscript $f$ denotes the "final" state.
Here, $\delta = \text{sign}(h_1 - h_3)$.
\end{corollary}

Once we have derived the exact relation between the input coherence and the unitarily extracted CCMW in Theorem 1, we would like to highlight two contrasting scenarios that explicitly demonstrate the dependence of CCMW on the choice of the coherence basis. We show how suitably choosing the coherence basis can both enhance and reduce CCMW with an increase in coherence. This marks one of the main results of our paper. In the subsequent section, we show that such a response exists even in higher dimensions. The two scenarios are presented in the form of two remarks below.

\textbf{Remark 1: } When the coherence is conserved in the energy eigenbasis, that is, the Hamiltonian is diagonal in the $\{\ket{0}, \ket{1}\}$ basis such that the off-diagonal element $h_2 = 0$, the only contribution comes from the first term of Eq.~\eqref{eq: xi2}, i.e., $\xi_2 = |h_{1} - h_{3}|\sqrt{1 - \mathcal{C}^2}$. Clearly, this is a monotonically decreasing function of coherence, suggesting that as coherence increases in the energy eigenbasis, the unitarily extracted CCMW always decreases. However, in the next remark presented below, we point out that with a suitable choice of coherence basis, one can also increase the unitarily extracted CCMW with increasing coherence.

\textbf{Remark 2:} When the Hamiltonian has non-zero off-diagonal elements and the diagonal elements are all zero or equal when written in the coherence basis, the contribution to the unitarily extracted CCMW comes only from the second term in Eq.~\eqref{eq: xi2}, i.e., $\xi_2 = 2h_{2}\mathcal{C}$. Clearly, in this case, $\xi_2$ increases linearly with coherence. This suggests that for the chosen coherence basis, one can achieve an enhancement in the unitarily extracted CCMW with increasing coherence. 

Note that the above two remarks refer to two extreme cases. In general, the unitarily extracted CCMW, as given in Eq.~\eqref{eq: xi2}, is the sum of two monotonic functions of coherence. The first term is a monotonically decreasing function of $\mathcal{C}$, arising from the diagonal part of the Hamiltonian. The off-diagonal part of the Hamiltonian contributes to the second term, which increases linearly with coherence. Therefore, in generic cases, there is a competition between the two terms, and no definitive statement can be made about the overall behavior of the unitarily extracted CCMW. This kind of behavior is also observed in higher-dimensional systems.

In the next section, we extend our discussion to higher-dimensional systems.

\section{Higher dimensional batteries} \label{sec: qudit_system}
 In this section, we consider higher-dimensional quantum batteries and analyze three distinct cases\\
 
\noindent
\textbf{Case 1:} The Hamiltonian is diagonal in the coherence basis.\\
\textbf{Case 2:} The Hamiltonian has off-diagonal elements, and all diagonal elements are either zero or equal when written in the coherence basis.\\
\textbf{Case 3:} The Hamiltonian contains both off-diagonal elements and unequal diagonal elements in the coherence basis.

We discuss \textbf{Case 1} in the following subsection.

 \subsection{When the Hamiltonian is Diagonal in Coherence Basis}\label{subsection: diagonal hamiltonian}

 In this subsection, we consider the Hamiltonian of qudit quantum batteries such that it is diagonal in the fixed basis of coherence. First, we numerically show that under such a scenario, the unitarily extracted CCMW corresponds to pure initial states. Next, considering pure initial states, we show that the CCMW always decreases with increasing coherence for arbitrary dimensions of the battery. Lastly, we present a closed-form relation between the CCMW and coherence for the case of a qutrit battery. Our analysis is presented below.

Consider a qudit battery with coherence fixed in the basis ${\ket{i}}$, where $i = 0, 1, \ldots, d-1$. We refer to this basis as the coherence basis. Let the Hamiltonian of the qudit battery in this basis be
\begin{equation}\label{eq: diagonal_hamiltonian}
J^{d}_{z} = \frac{2}{d - 1} \sum_{i = 0}^{d - 1} \left(i - \frac{d - 1}{2}\right) \ketbra{i}{i},
\end{equation}
with equispaced energy levels. Note that the choice of such a Hamiltonian is motivated by the fact that any qudit battery of dimension $d$ can be viewed as a spin-$s$ system, where $s = (d - 1)/2$. The Hamiltonian we consider is a $d$-dimensional generalization of the $z$-component of the angular momentum operator for a particle of spin-$s$. To ensure a meaningful comparison across different dimensions, we introduce a normalization factor of ${2}/({d - 1})$. This choice is deliberate because it ensures that, when the initial coherence is zero, the unitarily extracted CCMW corresponds exactly to the difference between the maximum and minimum eigenvalues of the Hamiltonian. By applying this normalization, we fix the maximum and minimum eigenvalues to $+1$ and $-1$, respectively independent of the dimension $d$. Such a normalization guarantees that the CCMW at zero coherence is always equal to $2$, regardless of the system’s dimension. This provides a consistent baseline across different battery dimensions. Starting from this common reference point, we then investigate how the CCMW varies as coherence increases, for various dimension.

Recall that, as discussed in Sec.~\ref{sec: ccmw}, the unitarily extracted CCMW for an arbitrary dimension $d$ at fixed coherence $\mathcal{C}$ is given by
\begin{equation}\label{eq: numerical_steps1}
\xi_{d}\left(\mathcal{C}\right) = \max_{\rho_{\text{in}}, \rho_{\text{f}}} \left( \Tr\left[\left(\rho_{\text{in}} - \rho_{\text{f}}\right) J^{d}_{z}\right] \right),
\end{equation}
such that the coherence of the initial and final states are equal, i.e., $C(\rho_{\text{in}}) = C(\rho_{\text{f}}) = \mathcal{C}$.
Here, the coherence for a given dimension $d$ varies in the range $\mathcal{C} \in \left[0, (d - 1) \right]$.
Thus, the problem of maximizing the CCMW is a constrained optimization problem. We perform this optimization numerically for quantum batteries of dimensions $d = 3, 4, 5, 6$, following the steps below

\begin{itemize}
    \item To construct $\rho_{\text{in}}$ and $\rho_{\text{f}}$, which are unitarily connected, we first consider states that are diagonal in the coherence basis, of the form
    \[
    \rho_{D} = \frac{\sum_{i = 0}^{d - 1} \lambda_i \ketbra{i}{i}}{\sum_{i = 0}^{d - 1} \lambda_i},
    \]
    with $\lambda_i \geq 0$ for all $i$. Clearly, the eigenvalues of $\rho_D$ are $\frac{\lambda_i}{\sum_{i} \lambda_i}$ for $i = 0, 1, \ldots, d - 1$. Here ``$D$" denotes diagonal.

    \item Next, we construct the states $\rho_{\text{in}}$ and $\rho_{\text{f}}$ from $\rho_D$ by applying two unitary operators, $U_i$ and $U_f$, respectively. This yields
    \[
    \rho_{\text{in}} = U_i \rho_D U_i^\dagger, \quad \rho_{\text{f}} = U_f \rho_D U_f^\dagger.
    \]
    Here, $U_i, U_f \in \mathbb{U}(d)$ are arbitrary unitary matrices of dimension $d$, each expressed as
\begin{align*}
U_{i/f} &= \exp\left(I \sum_{j = 0}^{d^2 - 1} \theta^j_{i/f} T_j\right),
\end{align*}
where the parameters satisfy $0 \leq \theta^j_{i/f} < 2\pi$ for all $j = 0, 1, \ldots, d^2 - 1$. The subscripts $i/f$ in $U_{i/f}$ and $\theta^j_{i/f}$ denote ``initial" and ``final", respectively.

    In this representation, $T_0=\mathbb{I}_d$ is the $d$-dimensional identity matrix, while the remaining $\{T_j\}_{j=1}^{d^2 - 1}$ are generators of the special unitary group $\mathbb{SU}(d)$. Our construction ensures that $\rho_{\text{in}}$ and $\rho_\text{f}$ are unitarily connected.

    \item Subsequently, we maximize $\Tr\left[\left(\rho_{\text{in}} - \rho_{\text{f}}\right) J^{d}_{z}\right]$ over the parameters of $\rho_D$, $U_i$, and $U_f$ to obtain the CCMW, ensuring that the coherence constraint
\[
C\left(\rho_{\text{in}}\right) = C\left(\rho_{\text{f}}\right) = \mathcal{C}
\]
is satisfied. We employ the ISRES algorithm of the Nonlinear Optimization (NLopt) library for this optimization.

\end{itemize}

After performing this numerical optimization, we observe two striking features, which we present as the following results:

\begin{description}
    \item[Result 1] For every coherence $\mathcal{C} \in \left[0, ({d - 1})/{2} \right]$, the initial state $\rho_{\text{in}}$ that delivers the CCMW is always a pure state.

    \item[Result 2] For every dimension $d = 3, 4, 5, 6$, the CCMW decreases with increase in coherence and eventually vanishes at maximum coherence.
\end{description}

\begin{figure}
\hspace{-0.8cm}
\includegraphics[scale=0.35]{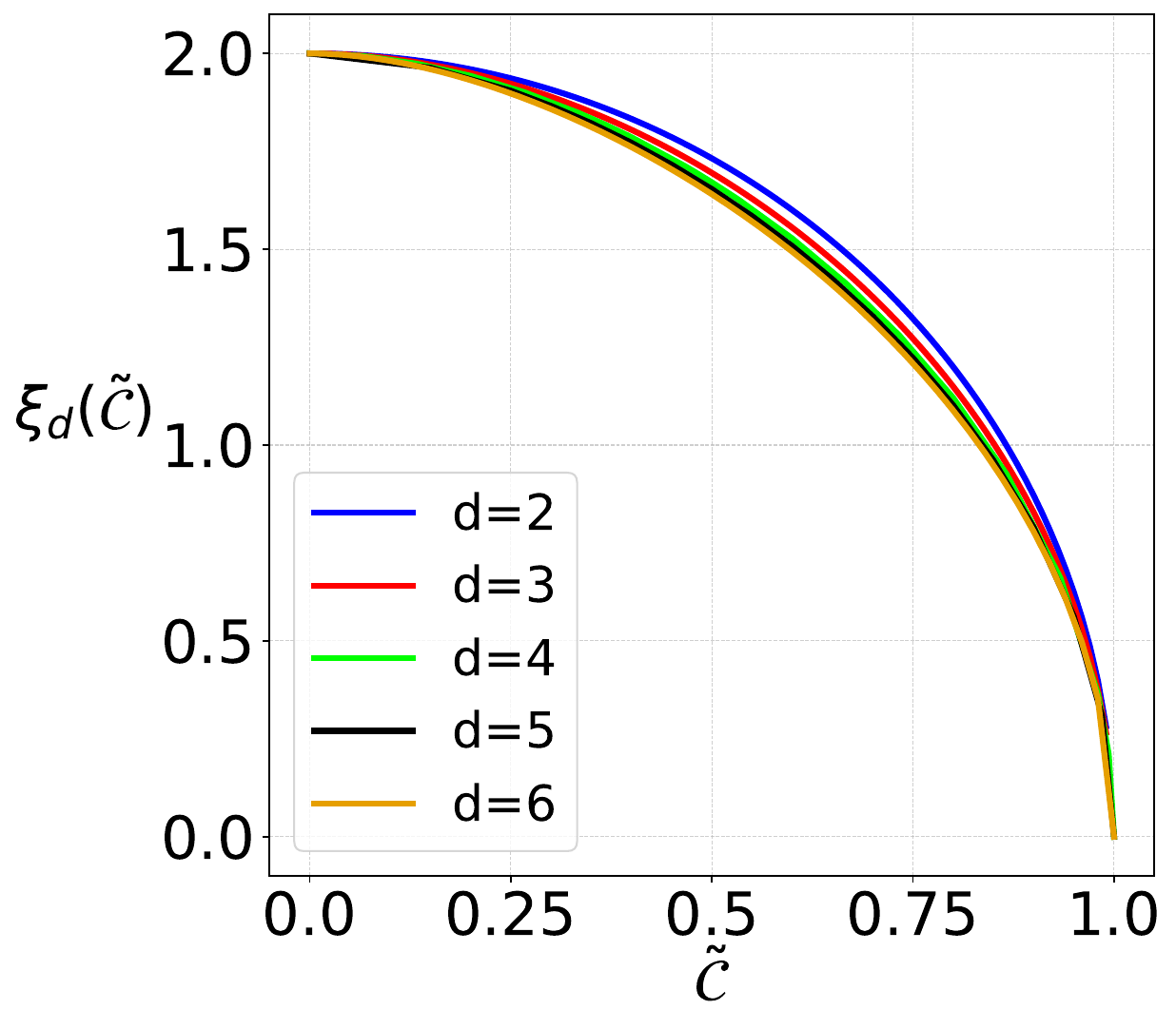}
    \caption{\textbf{Unitarily-extracted CCMW from qudit  battery with coherence, given in energy eigenbasis.} The plot depicts the variation of the CCMW, $\xi_d(\tilde{\mathcal{C}})$, with increasing scaled coherence $\tilde{\mathcal{C}}$ for battery dimensions $d = 2, 3, 4, 5, 6$, shown in different colors.  
It can be observed that, for all the considered dimensions, the CCMW decreases monotonically with increasing coherence.  
Thus, this plot illustrates that whenever coherence is fixed in the energy eigenbasis, the performance of the qudit battery deteriorates as coherence increases.  
The vertical axis in the plot has units of energy, whereas the horizontal axis is dimensionless.}
 \label{fig: ergotropy_diagonal}
\end{figure} 
To demonstrate Result 2 and Remark 1 in Fig.~\ref{fig: ergotropy_diagonal}, we plot the variation of the CCMW with coherence for $d = 2, 3, 4, 5, 6$.  
Since the maximum possible coherence of a $d$‑dimensional quantum battery is $\mathcal{C}_{\mathrm{max}}^{d} = (d - 1)/2$, the range of coherence values naturally depends on the battery dimension.  
This makes it difficult to compare the CCMW behavior across different dimensions on the same footing. To enable a dimension-independent and unified visualization, we rescale the coherence so that it always ranges from 0 to 1, regardless of $d$.  
This is achieved by dividing the coherence $\mathcal{C}$ by its maximum value $\mathcal{C}_{\mathrm{max}}^{d}$, giving the scaled coherence
\[
\tilde{\mathcal{C}} = \frac{\mathcal{C}}{\mathcal{C}_{\mathrm{max}}^{d}}.
\] 
Accordingly, the CCMW can now be viewed as a function of the scaled coherence, written as $\xi_{d}(\tilde{\mathcal{C}})$.
 
In Fig.~\ref{fig: ergotropy_diagonal} we plot $\xi_{d}(\tilde{\mathcal{C}})$ versus $\tilde{\mathcal{C}}$ for battery dimensions $d = 3, 4, 5, 6$.  
As seen in the plot, for all dimensions considered, $\xi_{d}(\tilde{\mathcal{C}})$ decreases monotonically with $\tilde{\mathcal{C}}$ and vanishes at maximum coherence, $\tilde{\mathcal{C}} = 1$.  
This confirms that when the coherence of the quantum battery is fixed in the Hamiltonian basis, the CCMW always deteriorates as the coherence increases.

Combining the observations in Results~1 and~2, we formulate the following theorem, which guarantees that for every battery dimension $d$, if the coherence is fixed in the Hamiltonian basis during the extraction of coherence‑constrained work, and the maximization is performed over only pure initial states, the extracted work invariably decreases with increasing coherence. The theorem is stated below.
    

\begin{theorem}
If only pure states with a given amount of coherence $\mathcal{C}$, fixed in the energy eigenbasis, are considered for the extraction of energy from a qudit battery of dimension $d$, under coherence-conserving unitaries $U_\mathcal{C}$, then the optimal work, maximized over all such unitaries and pure initial states, always decreases monotonically with increasing coherence $\mathcal{C} \in [0, d - 1]$, and eventually vanishes at maximum coherence $\mathcal{C} = d-1$.
\end{theorem}

\begin{proof}
Let the coherence of the qudit battery be fixed in the energy eigenbasis $\{\ket{i}\}$, where $i = 0, 1, 2, \ldots, d-1$. Suppose the Hamiltonian of the battery is given by
\begin{equation}
\label{PHD}
\bar{H}_d = \sum_i \epsilon_i^d \ketbra{i}{i}.
\end{equation}
Where for convenience, we have organized the basis $\{\ket{i}\}$ in such a way that $\epsilon_{i}^d\geq \epsilon_{j}^d$, for $i\leq j.$  We define  $\Delta \epsilon_{i,j}^d=\epsilon^{d}_{j}-\epsilon^{d}_{i}$.
Note that the Hamiltonian $J^{d}_{z}$ defined in Eq.~\eqref{eq: diagonal_hamiltonian} is a special case of $\bar{H}_d$, with $\epsilon^{d}_0 = -1$, $\epsilon^{d}_{d-1} = 1$, and $\Delta \epsilon^{d}_{i,j}= \Delta \epsilon^{d}_{l,m}$ for all $i,j \neq l,m$.

Let $\zeta^{\mathcal{C}}_d$ denote the set of all pure states $\rho_{\text{in}} = \ketbra{\psi_{\text{in}}^d}{\psi_{\text{in}}^d}$ of a $d$-dimensional quantum battery, where the coherence is fixed at $\mathcal{C}$. Energy is extracted from $\rho_{\text{in}}$ via coherence-conserving unitaries $U_\mathcal{C} \in \mathcal{U}_d^\mathcal{C}(\rho_{\text{in}})$, where $\mathcal{U}_d^\mathcal{C}(\rho_{\text{in}})$ denotes the set of all unitaries that conserve the coherence of $\rho_{\text{in}}$. Let the final state of the battery after energy extraction be $\rho_{\text{f}}$. 

Then, the optimal energy (work) extracted, maximized over all such choices of $\rho_{\text{in}}$ and corresponding $U_\mathcal{C}$, is given by
\begin{equation}
\xi_d^p(\mathcal{C}) = \max_{\rho_{\text{in}} \in \zeta^{\mathcal{C}}_d, \, U_\mathcal{C} \in \mathcal{U}_d^\mathcal{C}\left(\rho_{\text{in}}\right)} \Tr[\bar{H}_d (\rho_{\text{in}} - \rho_{\text{f}})].
\end{equation}
Here, the superscript $p$ in $\xi_d^p(\mathcal{C})$ denotes that the optimization is restricted to pure states with fixed coherence, $\mathcal{C}$. 

It is important to mention that the Result 1 suggest that for $d = 3,4,5,6$ and under the assumption $\bar{H}_d = J^{d}_z$, $\xi_d^p(\mathcal{C})$ coincides CCMW $\xi_d(\mathcal{C})$. However, for $d \geq 6$, and even for $d = 3,4,5,6$ if $\bar{H}_d$ is not equal to $J^{d}_z$, this identification may no longer hold.

Any arbitrary pure state $\ket{\psi_d}$ can be expressed in the basis $\{\ket{i}\}$ as 
\begin{equation}\label{eq: states}
\ket{\psi_d} = \sum_{i=0}^{d-1} x_i e^{i \theta_i} \ket{i}, \quad 0 \leq x_i \leq 1,
\end{equation}
with the normalization condition 
\begin{equation}\label{eq: normalization condition}
\sum_{i=0}^{d-1} x_i^2 = 1.
\end{equation}
If this state has coherence $\mathcal{C}$, then $\{x_i\}$ must also satisfy 
\begin{equation}\label{eq: fixed coherence condition1}
\sum_{i \neq j} x_i x_j = \mathcal{C}.
\end{equation}
Combining Eq.~\eqref{eq: normalization condition} with Eq.~\eqref{eq: fixed coherence condition1}, we can rewrite the coherence constraint as 
\begin{equation}\label{eq: fixed coherence condition2}
\left(\sum_{i=0}^{d-1} x_i \right)^2 = 1 + \mathcal{C} \quad \Rightarrow \quad \sum_{i=0}^{d-1} x_i = \sqrt{1 + \mathcal{C}}.
\end{equation}

So, any $d$-dimensional pure state written in the orthonormal coherence basis, as in Eq.~(\ref{eq: states}), must satisfy both Eq.~(\ref{eq: normalization condition}) and Eq.~(\ref{eq: fixed coherence condition2}). Geometrically, the points $\{x_{i}\}$ that simultaneously satisfy Eq.~(\ref{eq: normalization condition}) and Eq.~(\ref{eq: fixed coherence condition2}) must lie on the $d$-dimensional curve formed by the intersection of the higher-dimensional sphere given by Eq.~(\ref{eq: normalization condition}) and the higher-dimensional plane given by Eq.~(\ref{eq: fixed coherence condition2}). We can use the constraint in Eq.~\eqref{eq: fixed coherence condition2} to eliminate one of the coordinates, say $x_0$, from our analysis. Doing so gives the projection of the intersection curve between the sphere and the plane onto the subspace orthogonal to the $x_0$-axis. The compact equation that describes this projection is
\begin{equation}\label{eq: ellipse equation}
    \sum_{i=1}^{d-1}x_{i}^{2}+\sum_{i>j=1}^{d-2}x_{i}x_{j}-\sqrt{1+\mathcal{C}}\sum_{i=1}^{d-1}x_{i}+\frac{\mathcal{C}}{2}=0.
\end{equation}

This is an equation of an ellipse in the $(d-1)$ dimensional space with axes $\{x_{1},x_{2},...x_{d-1}\}$, for a fixed value of $\mathcal{C}\in[0,d-1]$. We call these higher-dimensional ellipses as isocoherent ellipses prior to the fact that every point on the ellipse corresponds to pure $d$ dimensional quantum states having equal level of coherence. 
It is worthy to note that the distance of the hyperplane in Eq.~\eqref{eq: fixed coherence condition2}, from the origin is $\sqrt{{(1+\mathcal{C})}/{d}}$. So, when coherence level increases from zero to $d-1$, the higher-dimensional plane becomes more distant from the origin. As a result of this, the area of the curve resulting from the intersection of the unit hypersphere in Eq.~\eqref{eq: normalization condition} and the higher-dimensional plane in Eq.~\eqref{eq: fixed coherence condition2} becomes smaller. As a result, the higher dimensional ellipse in Eq.~\eqref{eq: ellipse equation}, which is just the projection of the hypersphere-higher dimensional plane intersection curve, becomes smaller and smaller. And at maximum attainable coherence, the ellipse becomes a single point. This corresponds to the higher-dimensional plane in Eq.~(\ref{eq: fixed coherence condition2}) just touching the sphere of Eq.~(\ref{eq: normalization condition}) at a single point. As an illustration of the above discussion considering qutrit battery, we plot some of these isocoherent ellipses as shown in Fig.~\ref{fig: ellipses_ergotropy_diagonal}, for various values of coherence. As it can be seen in the Fig~\ref{fig: ellipses_ergotropy_diagonal} as the value of coherence increase the area of the coherence ellipses decreases, and eventually ceases to a point at maximum coherence $\mathcal{C}=2$.
\begin{figure}
\hspace{-1.5cm}
\includegraphics[scale=0.78]{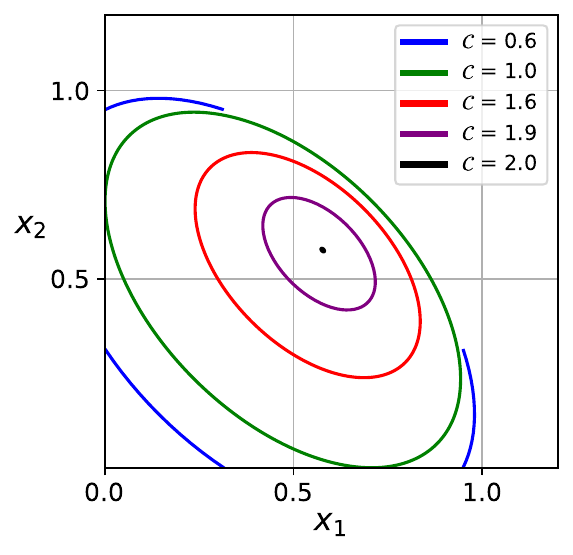}
   \caption{\textbf{Isocoherence ellipses for a qutrit system.} The figure shows isocoherent ellipses corresponding to several fixed coherence values \(\mathcal{C}\) for a qutrit system. Each ellipse is plotted in a different color according to \(\mathcal{C}\) and is defined by the state parameters \(x_{1}\), \(x_{2}\), and the fixed coherence \(\mathcal{C}\). As \(\mathcal{C}\) increases, the area of each ellipse decreases and eventually vanishes at maximum coherence \(\mathcal{C}=2\). Both axes in the plot are dimensionless.}

    \label{fig: ellipses_ergotropy_diagonal}
\end{figure}

Once we became familiar with the isocoherent curve. We use them to find  $\xi^{p}_{d}$ as follows.

Let the initial states with fixed coherence in the basis $\{\ket{i}\}$ be given as 
\begin{equation}\label{eq: pure_in_state}
    \ket{\psi_\text{in}^d}=\sum_{i=0}^{d-1}v_{i}e^{\theta_{i}}\ket{i}.
\end{equation}
 After energy extraction from this state under coherence-preserving unitaries, the final state is given as
\begin{equation}\label{eq: pure_out_state}
    \ket{\psi_\text{f}^d}\sum_{i=0}^{d-1}w_{i}e^{\phi_{i}}\ket{i}.
\end{equation} 

Note that both $\ket{\psi_\text{in}^d}$ and $\ket{\psi_\text{f}^d}$ are states of fixed coherence $\mathcal{C}$. Therefore the coordinates $\{v_i\}$ and $\{w_i\}$ with $i=1,2,\ldots,d-1$ must satisfy the equation of the projected higher dimensional ellipse as given in Eq.~\eqref{eq: ellipse equation}. Accordingly we have 
\begin{equation}\label{eq: with_delta_eq}
\xi^{p}_{d}\left(\mathcal{C}\right)=\max_{\{v_{i},v_{i}\}}\left[\sum_{i=1}^{d-1}\Delta \epsilon_{0,i}^d v_{i}^2-\sum_{i=1}^{d-1}\Delta \epsilon_{0,i}^d w_{i}^2\right].
\end{equation}
Where the maximization is performed over the coordinate sets $\{v_i\}$ and $\{w_i\}$, such that for $i=1,2,\ldots,d-1$ the points in each of the coordinate sets individually satisfies Eq.~\eqref{eq: ellipse equation}. 

 Note that in Eq.~\eqref{eq: with_delta_eq} the first term depends only on $\{v_i\}$ and second term depends on $\{w_i\}$. This gives us the the freedom to optimize the two terms individually. Thus we have 
 \begin{equation}
\xi^{p}_{d}\left(\mathcal{C}\right)=\max_{\{v_{i}\}}\left[\sum_{i=1}^{d-1}\Delta \epsilon_{0,i}^d v_{i}^2\right]-\min_{w_i}\left[\sum_{i=1}^{d-1}\Delta \epsilon_{0,i}^d w_{i}^2\right].
\end{equation}

To further simplify the expression, we perform a change of variables given by
\begin{equation}
X_{i} = \sqrt{\Delta \epsilon_{0,i}^d} \, v_{i}, \quad Y_{i} = \sqrt{\Delta \epsilon_{0,i}^d} \, w_{i}.
\end{equation}
After this transformation, the unitarily extracted CCMW takes the form
\begin{equation}\label{eq: scaled ergotropy}
    \xi^{p}_{d}\left(\mathcal{C}\right) = \max_{\{X_{i}\}} \sum_{i=1}^{d-1} X_{i}^2 - \min_{\{Y_{i}\}} \sum_{i=1}^{d-1} Y_{i}^2,
\end{equation}
where the sets \(\{X_{i}\}\) and \(\{Y_{i}\}\) satisfy the constraint
\small
\begin{equation*}
    \sum_{i=1}^{d-1} \frac{Z_{i}^{2}}{\Delta \epsilon_{0,i}^d}
    + \sum_{i>j=1}^{d-2} \frac{Z_{i} Z_{j}}{\sqrt{\Delta \epsilon_{0,i}^d \Delta \epsilon_{0,j}^d}}
    - \sqrt{1 + \mathcal{C}} \sum_{i=1}^{d-1} \frac{Z_{i}}{\sqrt{\Delta \epsilon_{0,i}^d}}
    + \frac{\mathcal{C}}{2} = 0,
\end{equation*}
\normalsize
with \(Z_i = X_i\) or \(Y_i\).

The above constraint defines a \emph{scaled ellipse}, corresponding to each equal-coherence ellipse described in Eq.~\eqref{eq: ellipse equation}. Importantly, the scaling preserves the qualitative features of the original ellipses specifically, the area of the scaled ellipses shrinks with increasing \(\mathcal{C}\), and vanishes at maximum coherence.

We observe that the sets $\{X_{i}\}$ and $\{Y_{i}\}$ represent two points on the same ellipse, since the corresponding coherence values are identical. Thus, the coordinate transformation implies that $\xi^{p}_{d}(\mathcal{C})$, for a fixed $\mathcal{C}$, as given in Eq.~(\ref{eq: scaled ergotropy}), corresponds to the difference between the squares of two optimized Euclidean distances in a $(d{-}1)$-dimensional space. The first is the \emph{maximum} distance from the origin $(X_i, Y_i = 0\,\, \forall i)$ to a point on the scaled equal-coherence ellipse corresponding to coherence $\mathcal{C}$, and the second is the \emph{minimum} such distance. From our earlier discussion, we know that as the coherence $\mathcal{C}$ increases, the area of the corresponding scaled equal-coherence ellipse shrinks. Consequently, as the ellipse contracts with increasing $\mathcal{C}$, the difference between these optimized distances and thus $\xi^{p}_{d}(\mathcal{C})$ decreases. At the point of maximum coherence, the ellipse collapses into a single point, and this difference becomes zero. This can also be understood from the fact that there exists only one pure state of maximum coherence. Therefore, the initial and final states in this scenario are identical, implying that no energy can be extracted when $\mathcal{C} = d - 1$. This completes the proof of Theorem~2.
\end{proof} 

Next, considering the Hamiltonian $J^{d}_z$ of the form given in Eq.~\eqref{eq: diagonal_hamiltonian}, and building upon the framework developed in Theorem 2, we derive a closed form expression for $\xi_d^p(\mathcal{C})$ which, in this context, is equivalent to the CCMW as a function of fixed coherence $\mathcal{C} \in [0,2]$. While the same procedure can, in principle, be extended to higher dimensions ($d > 3$), the analytical complexity increases significantly with $d$. To gain concrete insights and maintain analytical tractability, we focus here on the qutrit ($d = 3$) case.  The exact steps and the expression for CCMW for the qutrit case is presented in the theorem below.
\begin{theorem}
    Considering pure qutrit battery of fixed coherence $\mathcal{C}\in[0,2]$ and the Hamiltonian of the form given in Eq.~\eqref{eq: diagonal_hamiltonian} the CCMW $\xi_3(\mathcal{C})$ can be expressed in terms of fixed coherence $\mathcal{C}$ as
    \begin{equation*}
\xi_{3}(\mathcal{C})=\sqrt{\left(1+f_{3}+\frac{\mathcal{C}}{3}\right)\left(1+f_{3}-\mathcal{C}\right)},
\end{equation*}
where $f_{3}=\sqrt{1+\mathcal{C}}\sqrt{1-\frac{\mathcal{C}}{3}}$.
\end{theorem}
\begin{proof}
Before we begin the proof we would like to recall that for $d=3$, and for diagonal Hamiltonian of the form given in Eq.~\eqref{eq: diagonal_hamiltonian}, one have $\xi^p_{3}=\xi_{3}$. From the proof of Theorem~1, we know that for an arbitrary dimension $d$, the quantity $\xi_d^p(\mathcal{C})$ can be obtained by computing the difference between the distances from the origin to the closest and farthest points on the scaled $(d-1)$-dimensional higher dimensional ellipse. In the case of $d = 3$, this higher dimensional ellipse reduces to a two-dimensional ellipse.

In such a case the minimum (maximum) distance will be radius of a circle which just touched by the ellipse at a point from outside (inside) . Let the corresponding points where the circle just touched the ellipse be $(\tilde{x},\tilde{y})$. At this point, the unit normal vector to the both ellipse and the circle will be same. Using this fact, we can readily get an equation for $(\tilde{x},\tilde{y})$ as
\begin{equation}\label{eq: touched_outside}
    \frac{\tilde{y}}{\tilde{x}}=\frac{2\tilde{y}+\frac{\tilde{x}}{\sqrt{2}}-\sqrt{1+\mathcal{C}}}{\tilde{x}+\frac{\tilde{y}}{\sqrt{2}}-\sqrt{\frac{1+\mathcal{C}}{2}}}.
\end{equation}
Now, $(\tilde{x},\tilde{y})$ lies on the scaled ellipse, so it must satisfy the scaled ellipse equation. Thus we have
\begin{equation}\label{eq: satisfy_ellipse}
    \frac{\tilde{x}^{2}}{2}+\tilde{y}^{2}+\frac{\tilde{x}\tilde{y}}{\sqrt{2}}-\sqrt{1+\mathcal{C}}\left(\frac{\tilde{x}}{\sqrt{2}}+\tilde{y}\right)+\frac{\mathcal{C}}{2}=0.
\end{equation}
From these two equations, one  can eliminate $\tilde{x}$ and find solution corresponding to $\tilde{y}$ as $$\tilde{y}_0=\frac{\sqrt{1+\mathcal{C}}-\sqrt{1-\mathcal{C}/3}}{2},$$ 
Note that in deriving this we ignored the unphysical solution corresponding to $\tilde{y}\geq 1$ . Putting $\tilde{y}_0$ in Eq.~\eqref{eq: satisfy_ellipse}, one can get two solutions of $\tilde{x}$ as
$$\tilde{x}_{\pm}=\frac{\sqrt{1+\mathcal{C}}+\sqrt{1-\mathcal{C}/3}}{2\sqrt{2}}\pm\sqrt{\frac{1-\mathcal{C}+\sqrt{1+\mathcal{C}}\sqrt{1-\mathcal{C}/3}}{4}}.$$

It is easy to check that the maximum distance from the origin to the scaled ellipse corresponds to the point $(\tilde{x}_{+},\tilde{y}_0)$ and the minimum distance corresponds to the point $(\tilde{x}_{-},\tilde{y}_0)$. 
Thus we have the CCMW as
\begin{equation*}
\begin{split}
\xi_{3}&=(\tilde{x}_{+}^{2}+\tilde{y}^{2}_0)-(\tilde{x}_{-}^{2}+\tilde{y}^{2}_0)\\
&=\sqrt{\left(1+f_{3}+\frac{\mathcal{C}}{3}\right)\left(1+f_{3}-\mathcal{C}\right)}.
    \end{split}
\end{equation*}
This completes the proof of Theorem 3.
\end{proof}

Our numerical analysis also reveals that the coherence-conserving unitary achieving the CCMW is state-independent; it preserves the coherence of every state in a given dimension. Motivated by this discovery, in the next section we investigate which states, under these universal coherence-preserving unitaries yield zero extractable energy when coherence is fixed in the energy basis. The detailed analysis and classification of these
``isocoherent passive” states are presented in the subsection below.

\subsubsection{Passive state for diagonal Hamiltonian}\label{sec: passive_state}
 In this section we discuss the form of \emph{isocoherent passive} states. We define an isocoherent passive state \(\rho_p\) in the Hilbert space \(\mathcal{H}_d\) as an initial state of a quantum battery that has fixed coherence $\mathcal{C}$ in the energy eigenbasis and from which no energy can be extracted by coherence-preserving unitary operations of the form

\begin{equation}\label{eq: coherence_preserving_all}
    \bar{U}_{\mathcal{C}} = \sum_{i=0}^{d-1} e^{\mathrm{i}\,\omega_i}\,\ketbra{p(i)}{i},
\end{equation}
in the coherence basis \(\{\ket{i}\}\), where \(p(i)\) denotes a permutation of the basis indices, and each \(\omega_i\) ranges over \([0,2\pi)\).

Note that such unitaries are independent of the input state. In other words, they preserve the coherence of all initial states within a given dimension. In our numerical analysis (Sec.~\ref{subsection: diagonal hamiltonian}), we observed that the optimal unitary achieving the CCMW exactly matches the form presented in Eq.~\eqref{eq: coherence_preserving_all}. This inspired us to restrict attention to this class of coherence‑preserving unitaries and to search for states of fixed coherence from which no energy can be extracted when using those unitaries. Below we derive the form of these ``incoherent passive'' states.

Let the Hamiltonian of the \(d\)-dimensional battery be \(\bar{H}_d\), as defined in Eq.~\eqref{PHD}. Now consider the quantities \(\Tr\bigl[\bar{H}_d\,\rho_p\bigr]\) and \(\Tr\bigl[\bar{H}_d\,\bar{U}_c\,\rho_p\,\bar{U}_c^\dagger\bigr]\). We can write:

\[
\Tr\bigl[\bar{H}_d\,\rho_p\bigr] = N\,\Tr\bigl[h_d\,\rho_p\bigr] + \epsilon_0^d,
\]

\[
\Tr\bigl[\bar{H}_d\,\bar{U}_c\,\rho_p\,\bar{U}_c^\dagger\bigr] = N\,\Tr\bigl[\bar{U}_c^\dagger\,h_d\,\bar{U}_c\,\rho_p\bigr] + \epsilon_0^d.
\]

Here we  define
$$h_{d}=\frac{\bar{H}_{d}-\epsilon^{d}_0\mathbb{I}_{d}}{N}, \quad N=\Tr\left[\bar{H}_{d}-\epsilon^{d}_0\mathbb{I}_{d}\right].$$ 
Where $\epsilon_0$ denotes the minimum energy of the Hamiltonian $\bar{H}_{d}$ and $\mathbb{I}_{d}$ denotes $d$ dimensional identity matrix. Note $h_{d}$ has all the properties of a density matrix. Furthermore if $\bar{U}_{c}$ preserves coherence of any state, so does $\bar{U}_{c}^{\dag}$. So $\bar{U}_{c}^{\dag}h_{d} \bar{U}_{c}$ has same coherence as $h_{d}$. But $h_{d}$ is diagonal and the coherence basis is the eigenbasis of the Hamiltonian. So, $\bar{U}_{c}^{\dag}h_{d} \bar{U}_{c}$ is the same as $h_{d}$ with the diagonal entries just permuted. Then $$\Tr\left[\bar{H}_{d}\bar{U}_{c} \rho_{p}\bar{U}_{c}^{\dag} \right]=\epsilon^{d}_0+N\sum_{i=1}^{d}\left[h_{d}\right]_{p(i)}\left[\rho_{p}\right]_{i}.$$ Here $\left[h_{d}\right]_{i}$ is $i$-th diagonal element of $h_{d}$, $\left[\rho_{p}\right]_{i}$ is $i$-th diagonal element of $\rho_{p}$ and $p(i)$ indicates a permutation of $i$-th index. So $\rho_{p}$ is passive if, for all $\bar{U}_{\mathcal{C}}$, $\rho_{p}$ must satisfy 

$$\operatorname{Tr}[\rho_{p} \bar{H}_{d}] \le \operatorname{Tr}[\bar{U}_{\mathcal{C}}\,\rho_{p}\,\bar{U}_{\mathcal{C}}^{\dagger} \bar{H}_{d}]$$ which implies 

$$\sum_{i=1}^{d}\left[h_{d}\right]_{i}\left[\rho_{p}\right]_{i}\leq \sum_{i=1}^{d}\left[h_{d}\right]_{p(i)}\left[\rho_{p}\right]_{i}$$ 
for any permutations of indices. This is possible when the diagonal elements of \(\rho_p\) and those of \(h_d\) are ordered in opposite directions (see Theorem 368 in~\cite{hardy1934inequalities}). Specifically, if 
$
[h_d]_{i} \ge [h_d]_{j}$,
then \(\rho_p\) must satisfy
$[\rho_p]_{i} \le [\rho_p]_{j}$.

So, for a given Hamiltonian \(\bar{H}_d\) and coherence-preserving unitaries \(\bar{U}_{\mathcal{C}}\), which conserve coherence of any state in the energy eigenbasis, the form of the isocoherent passive state \(\rho_p\) must satisfy the following condition: when expressed in the energy eigenbasis, the diagonal elements of \(\rho_p\) must be ordered oppositely to those of the Hamiltonian. That is, higher-energy levels must have smaller or equal populations than lower-energy levels.

Equivalently, if the energy levels satisfy \(\epsilon^{d}_i \le \epsilon^{d}_j\), then for isocoherent passive states we require $
[\rho_p]_{i} \ge [\rho_p]_{j}.
$
This ensures no energy can be extracted from \(\rho_p\) under coherence-preserving unitary of the form \(\bar{U}_{\mathcal{C}}\).

Thus, our analysis of higher-dimensional systems with coherence fixed in the energy eigenbasis is complete. In the next section, we move beyond the energy eigenbasis and explore energy extraction from higher-dimensional systems when coherence is defined in a different basis. A detailed analysis is presented below.

\subsection{When the Hamiltonian has only Off-diagonal Elements in Coherence Basis} \label{sec: off_diagonal_hamiltonian}
In the previous sections, we have restricted ourselves to the Hamiltonians that are diagonal in the basis in which the coherence is given. In this section, we will discuss cases involving Hamiltonians with off-diagonal elements for higher-dimensional systems. 

As in previous sections, consider the coherence basis is $\ket{i}$, for $i=0,1,2,\ldots (d-1)$. For the sake of calculational simplicity, we take the Hamiltonian 
\begin{equation}\label{eq: non-diagonal_hamiltonian}
    J^{d}=\alpha_{1} J^{d}_{x}+\alpha_{2} J^{d}_{y}.
\end{equation}
Here $J^{d}_{x}$, $J^{d}_{y}$ is the generalization of the angular momentum component in $d-$ dimension. When written in the coherence basis $J^{d}_{x}$, $J^{d}_{y}$ have the form
\begin{align}\label{eq: jx_and_jy}
    J_x^d&=\sum_{i=1}^{d-1}\left(\ketbra{i}{i-1}+\ketbra{i-1}{i}\right) \notag \\
    J_y^d&=\sum_{i=1}^{d-1}I\left(\ketbra{i}{i-1}-\ketbra{i-1}{i}\right). \notag \\
\end{align}
Clearly, $J^{d}_{x}$, $J^{d}_{y}$ have only off-diagonal elements in the coherence basis. This $J^{d}$ can also be written as
\begin{equation}\label{eq: Jd}
    J^{d}=\alpha\sum_{i=1}^{d-1}\left[\exp\left(I\phi\right)\ketbra{i}{i-1}+h.c.\right],
\end{equation}
where $\alpha=\sqrt{\alpha^2_1+\alpha^2_2}$ and $\phi=\arg (\alpha_1+I\alpha_2)$.
Thus $J^d$ is suitable for our analysis considering Hamiltonian having only off-diagonal elements in the coherence basis.

 Note, for $d=2$, described in Corollary 1, $\xi^p_{2}(\mathcal{C})=\xi_{2}(\mathcal{C})$. Now, $\xi^p_{2}(\mathcal{C})$ have the exact form as given in Eq.~\eqref{eq: xi2}, with the diagonal elements $h_1=h_3=0$ of $J^2$ and its off-diagonal element is $h_2=\alpha$, in the coherence basis. We then analyze the qutrit case. For the qutrit battery, using the Hamiltonian in Eq.~\eqref{eq: non-diagonal_hamiltonian} with $d=3$, we have employed the exact same method, as discussed in the subsection~\ref{subsection: diagonal hamiltonian} for numerically evaluating unitarily extracted CCMW, using the ISRES algorithm of NLopt. Our numerical analysis revealed that even for this case considering pure states is sufficient for obtaining the unitarily extracted CCMW, i.e $\xi_{3}\left(\mathcal{C}\right)=\xi_{3}^p\left(\mathcal{C}\right)$. Equipped with this fact, we give a closed-form expression of unitarily extracted CCMW as follows.

\begin{figure}
    \centering
    \hspace{-1.5cm}
    \includegraphics[scale=0.4]{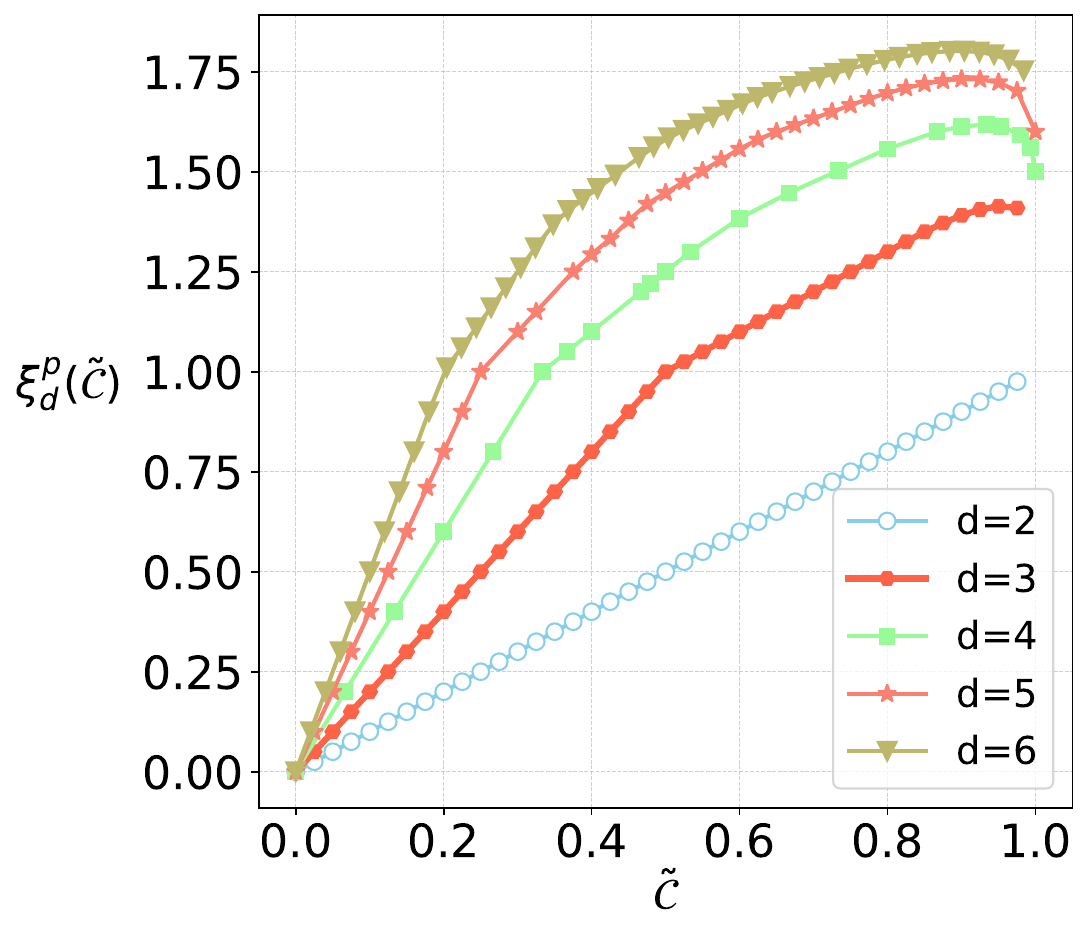}
    \caption{\textbf{Variation of $\xi^p_{d}(\mathcal{\tilde{C}})$ with fixed value of scaled coherence $\tilde{C}$, for qudit systems}. The figure depicts the variation of $\xi^p_{d}(\mathcal{\tilde{C}})$ with the scaled coherence $\mathcal{\tilde{C}}\in[0,1]$, for battery dimension $d=2,3,4,5,6$, plotted in different colors. As it can be seen in the plot that for $d=2$, $\xi^p_{d}(\mathcal{\tilde{C}})$ increases with increase in $\mathcal{\tilde{C}}\in[0,1]$. For $d \geq 3$, the plots depicts a generic feature exhibiting that $\xi^p_{d}(\mathcal{\tilde{C}})$ mostly increases with $\mathcal{\tilde{C}}$, followed by a slight decrease as the coherence reaches it's maximum value. The vertical axis are plotted in units of $\alpha$ which has dimension of energy and the horizontal axis is dimensionless.}

    \label{fig: jz0ergotropy}
\end{figure}

\begin{equation}\label{eq: closed_non-diagonal_qutrit}
    \xi_{3}\left(\mathcal{C}\right)=
    \begin{cases}
    2\alpha\mathcal{C}, & \text{if } 0 \leq \mathcal{C}\leq 1 \\
    \alpha(\mathcal{C}+1), & \text{if } 1 \leq \mathcal{C}\leq \frac{5}{3} \\
    2\alpha\left(\mathcal{C}-f(\mathcal{C})\right) & \text{if } \frac{5}{3} \leq \mathcal{C}\leq 2.
    \end{cases}
\end{equation}
Here $$ f(\mathcal{C})=\frac{2}{9}\left(\sqrt{1+\mathcal{C}}-\sqrt{1-\frac{\mathcal{C}}{2}}\right)^{2}.$$
The detail derivation of this expression is given in Appendix~\ref{appendix: non-diagonal qutrit}. 

Following the same argument as in Sec.~\ref{subsection: diagonal hamiltonian}, for dimension independent visualization we rescale the coherence to \(\tilde{\mathcal{C}}\), such that \(\tilde{\mathcal{C}} \in [0,1]\). We plot  $\xi_{3}(\tilde{\mathcal{C}})$, which is equivalent to $\xi_{3}^p(\tilde{\mathcal{C}})$, in units of $\alpha$ for various values of the scaled coherence $\tilde{\mathcal{C}}\in[0,1]$ in Fig.~\ref{fig: jz0ergotropy} (in tomato red). As the plot makes clear, the CCMW generally increases with coherence, with only a small decrease near the maximum coherence value. This indicates that when the Hamiltonian is such that it has only off-diagonal elements in coherence basis, the CCMW can enhance with increase in coherence which is in contrast to that of the diagonal Hamiltonian case discussed in Sec.~\ref{subsection: diagonal hamiltonian}. Our numerical analysis also reveals that, unlike the case of a diagonal Hamiltonian, the phases of the off‑diagonal elements in the initial state play a crucial role in determining the CCMW.

\begin{figure*}[t]
    \centering
    \begin{minipage}[b]{0.245\textwidth}
         \centering
         \includegraphics[width=\textwidth]{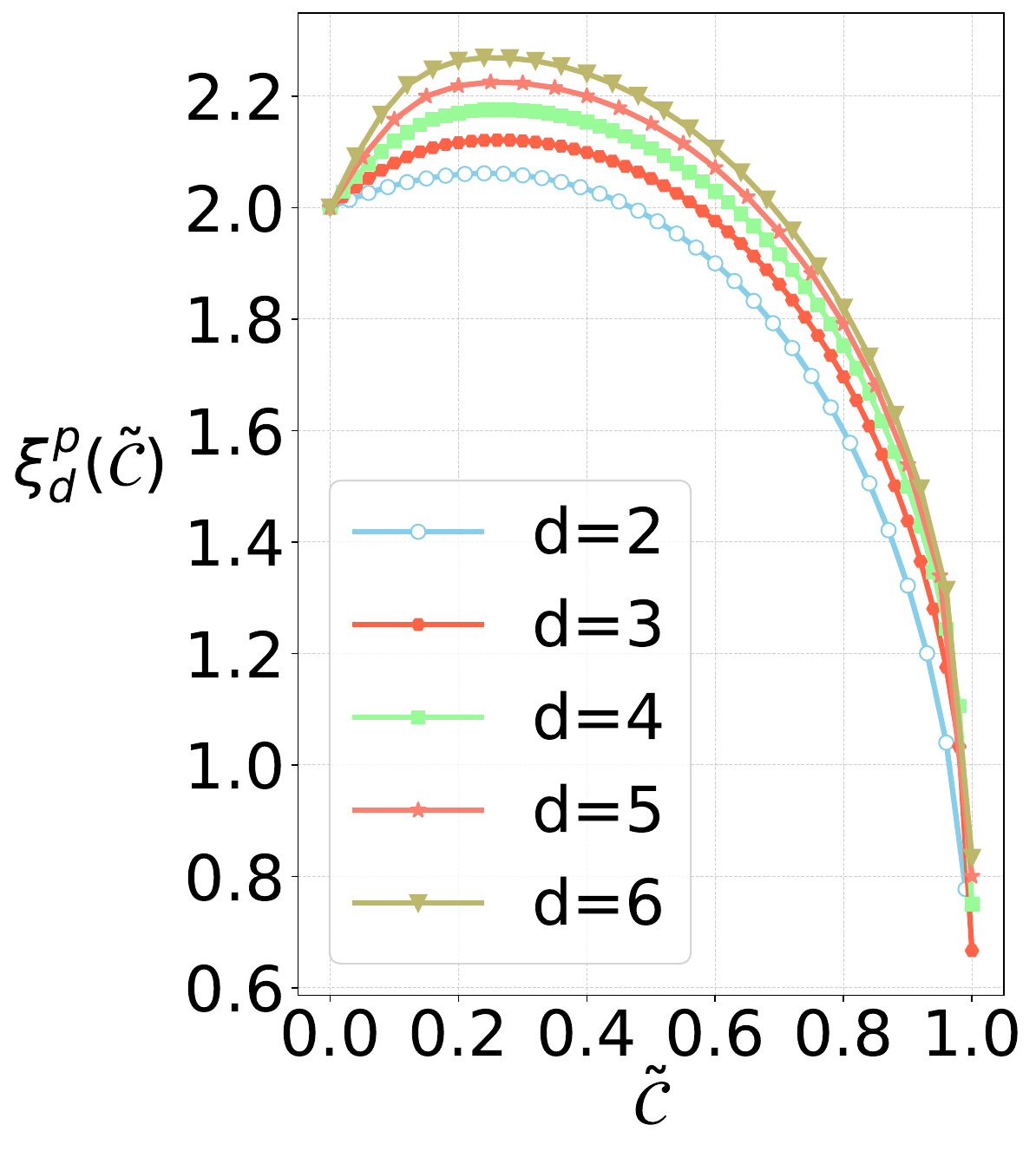}
     \end{minipage}
     \hfill
     \begin{minipage}[b]{0.245\textwidth}
         \centering
         \includegraphics[width=\textwidth]{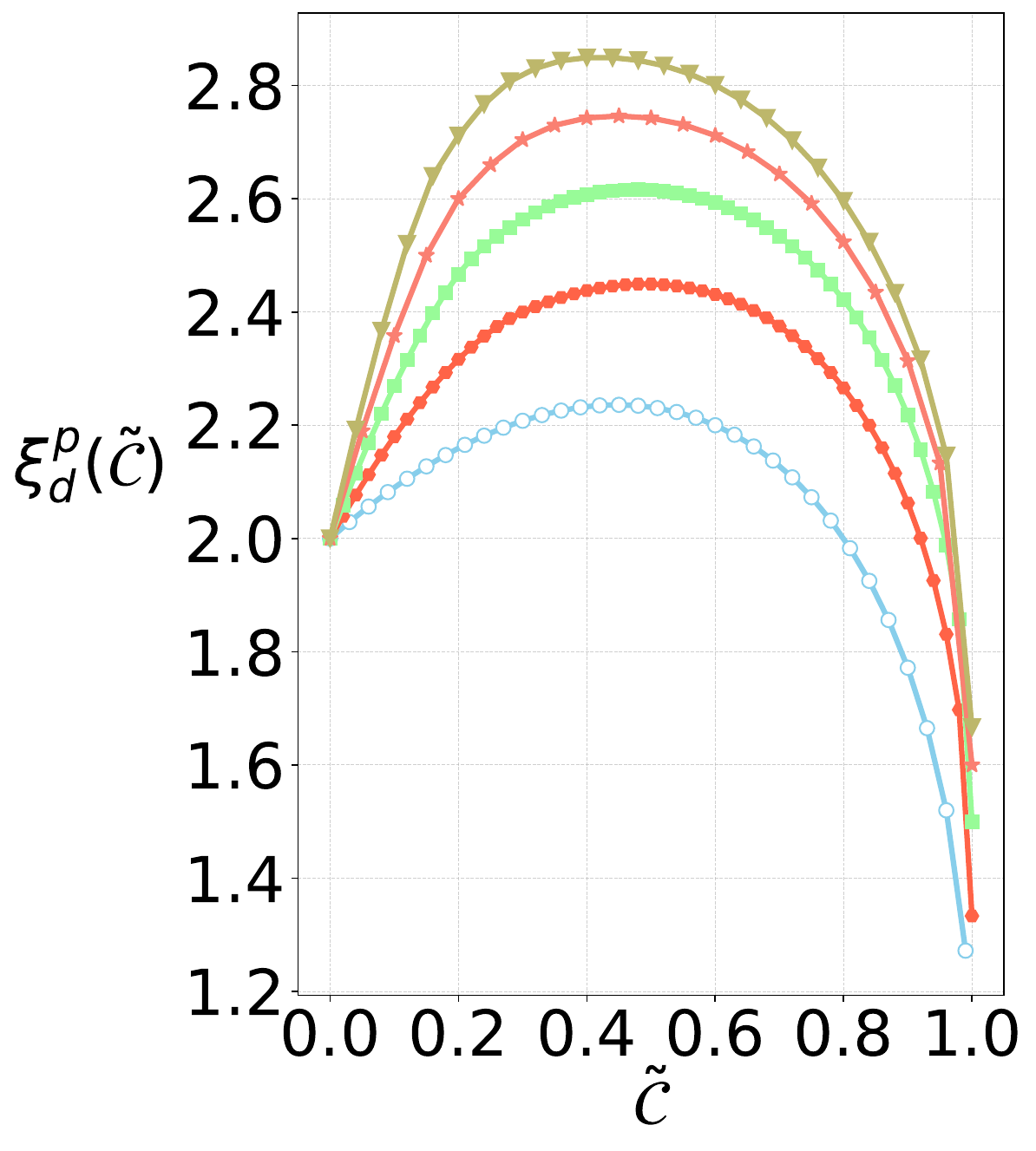}
     \end{minipage}
     \hfill
     \begin{minipage}[b]{0.245\textwidth}
         \centering
         \includegraphics[width=\textwidth]{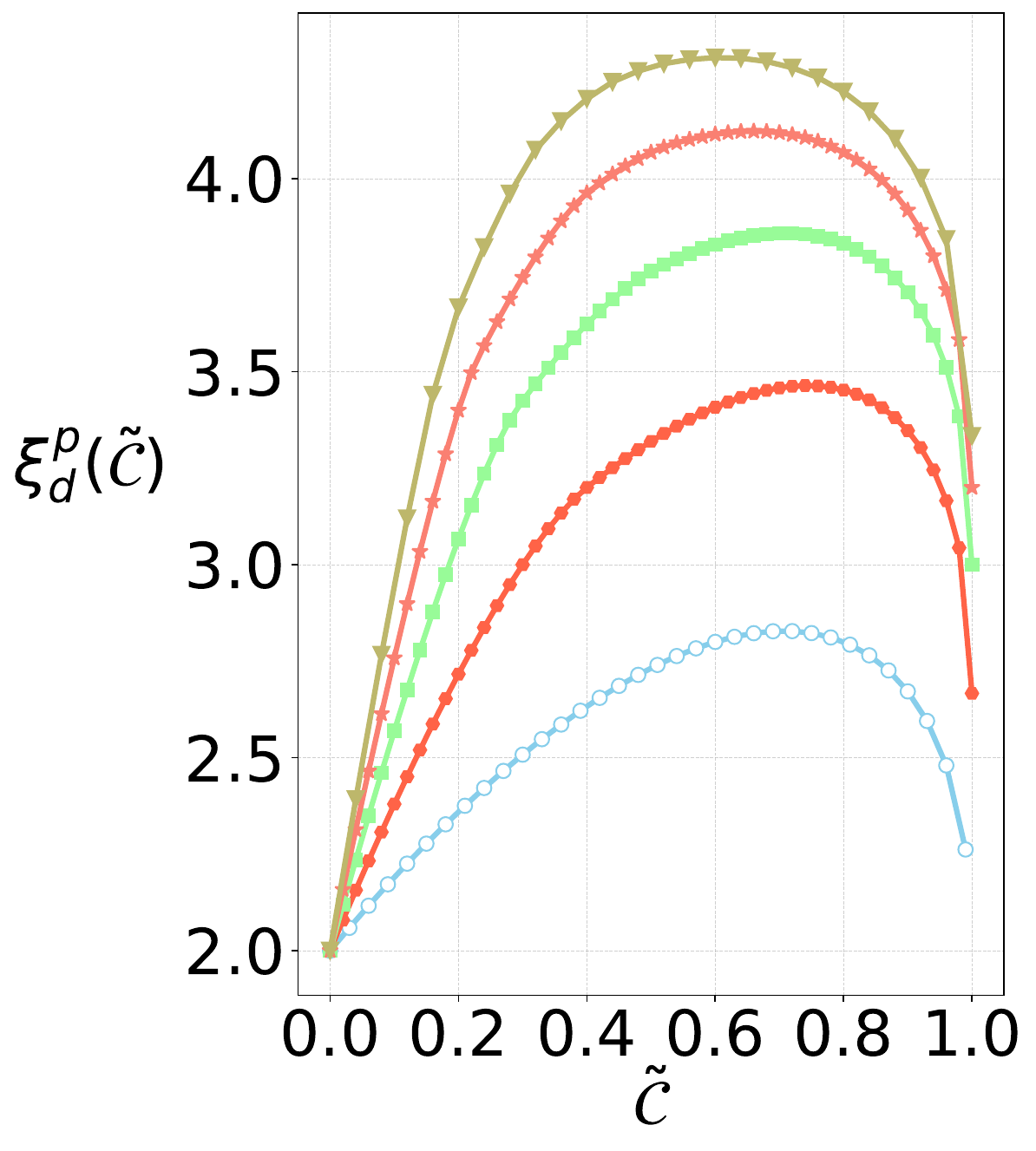}
     \end{minipage}
     \hfill
     \begin{minipage}[b]{0.245\textwidth}
         \centering
         \includegraphics[width=\textwidth]{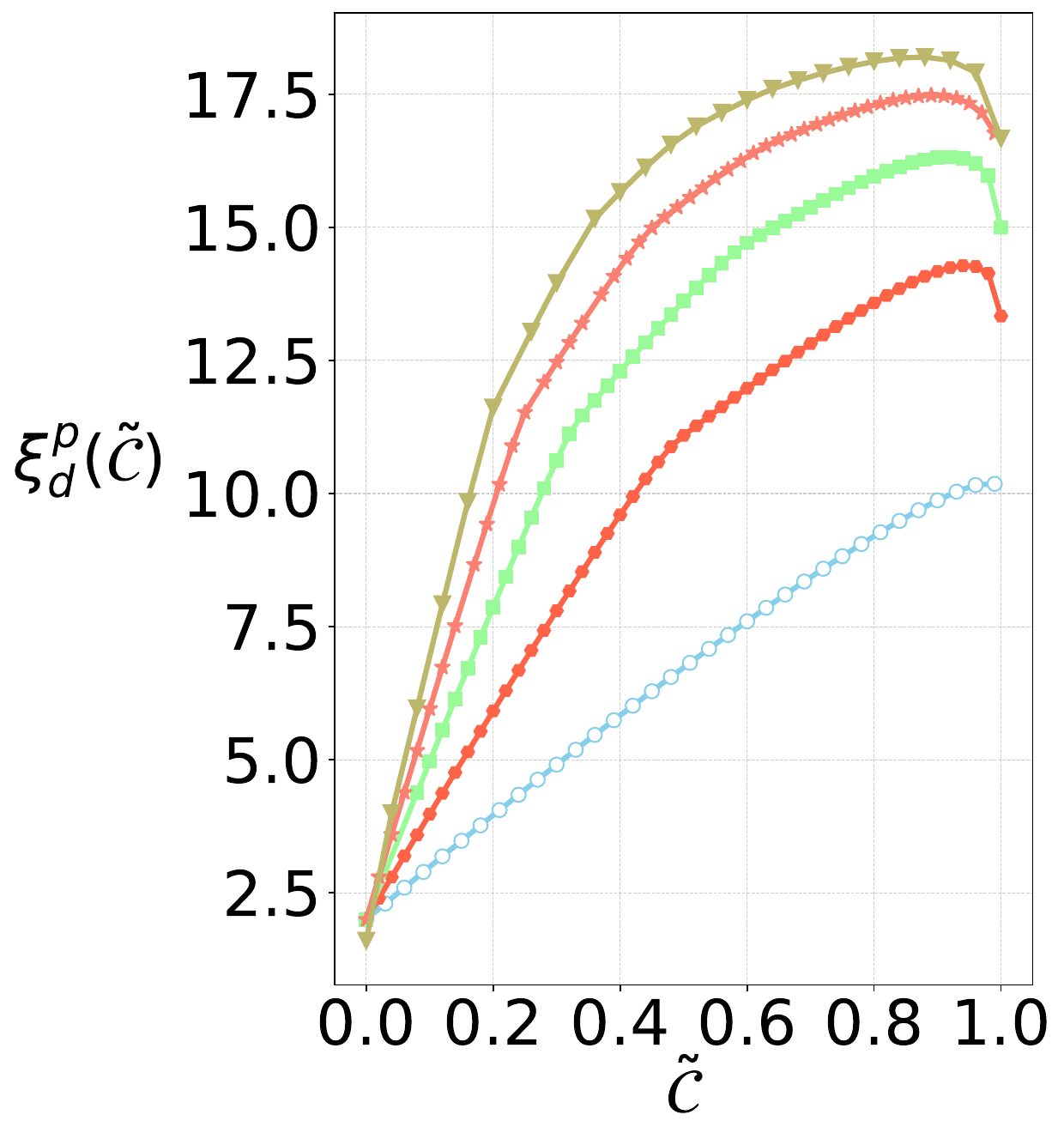}
     \end{minipage}
     \caption{\textbf{Coherence dependence of \(\xi^p_{d}(\tilde{\mathcal{C}})\) for various \(\bar{\alpha}\) values.} The figures show how \(\xi^p_{d}(\tilde{\mathcal{C}})\) varies with the scaled coherence \(\tilde{\mathcal{C}} \in [0,1]\) for dimensions \(d = 2, 3, 4, 5, 6\), with coherence fixed in the basis where the Hamiltonian is given by Eq.~\eqref{off}. The four panels correspond to \(\bar{\alpha} = 0.25, 0.5, 1,\) and \(5\). For small \(\bar{\alpha}\) (0.25 and 0.5), shown in the left panels, \(\xi^p_{d}(\tilde{\mathcal{C}})\) decreases as coherence increases. In contrast, for larger \(\bar{\alpha}\) values (1 and 5), shown in the right panels, \(\xi^p_{d}(\tilde{\mathcal{C}})\) increases with coherence. This left-to-right transition reflects the increasing contribution of off-diagonal elements in the Hamiltonian as \(\bar{\alpha}\) increases, shifting the system's behavior from that of a Hamiltonian diagonal in the coherence basis to one dominated by off-diagonal structure. The vertical axis is in units of energy, and the horizontal axis is dimensionless.}

    \label{fig: mixdiagonals}
\end{figure*}

Since so far our analysis shows that restricting the calculation to pure states suffices to attain the CCMW, we continue using this restriction to reduce numerical complexity, which grows with dimension. Therefore, for higher dimensions, specifically $d=4,5,6$, we numerically evaluate $\xi^p_{d}(\mathcal{C})$ as follows

We consider initial and the final states as in Eq.~\eqref{eq: pure_in_state} and Eq.~\eqref{eq: pure_out_state}. Next we  maximized $$\bra{\psi^{d}_{\text{in}}}J^{d}\ket{\psi^{d}_{\text{in}}}-\bra{\psi^{d}_{\text{f}}}J^{d}\ket{\psi^{d}_{\text{f}}}$$ over $\{v_{i},\theta_i\}$ and $\{w_{i},\phi_i\}$ such that both $\{v_{i}\}$ and $\{w_{i}\}$ independently satisfy the normalization condition in Eq.~\eqref{eq: normalization condition} and the coherence condition in Eq.~\eqref{eq: fixed coherence condition1}. We perform the maximization numerically. The plot of $\xi^{p}_{d}\left(\mathcal{C}\right)$ with the scaled coherence $\tilde{\mathcal{C}}$ is presented in Fig.~\ref{fig: jz0ergotropy}, for $d=\{4,5,6\}$. From the plot, it can be seen that $\xi^{p}_{d}\left(\tilde{\mathcal{C}}\right)$ increases mostly with coherence and exhibits the same behavior for the $d = 3$ case given in Eq.~\eqref{eq: closed_non-diagonal_qutrit}.

Since in general $\xi^{p}_{d}(\tilde{\mathcal{C}})$ provides a valid lower bound on the CCMW $\xi_{d}(\tilde{\mathcal{C}})$, an increase in $\xi^{p}_{d}(\tilde{\mathcal{C}})$ with $\tilde{\mathcal{C}}$ suggests that the CCMW may also increase with coherence. This behavior stands in stark contrast to the case where coherence is fixed in the Hamiltonian’s eigenbasis, in which the CCMW strictly decreases with increasing coherence. Our numerical analysis thus reveals a basis-dependent response of the CCMW. Specifically, when coherence is fixed in a basis where the Hamiltonian exhibits off-diagonal elements, the CCMW can increase with increasing coherence, unlike the case with energy-basis coherence.


In the next section we explore yet another scenario when the Hamiltonian has both diagonal and off-diagonal elements in the coherence basis. The detail analysis is presented below.

\subsection{When the Hamiltonian has Both Diagonal and  Off-diagonal Elements in  Coherence Basis} \label{sec: non-diagonal_&_diagonal}

In this section we analyze the case in which the Hamiltonian includes both diagonal and off-diagonal elements in the coherence basis \(\{\ket{i}\}\). Specifically, we consider the Hamiltonian
\[
\bar{J}^{d} = \alpha_{1} J^{d}_{x} + \alpha_{2} J^{d}_{y} + \alpha_{3} J^{d}_{z},
\]
where \(J^{d}_x\) and \(J^{d}_y\) are defined in Eq.~\eqref{eq: jx_and_jy}, and
\[
J^{d}_z = \frac{2}{d-1}\sum_{i=0}^{d-1}\left(i - \frac{d-1}{2}\right)\ketbra{i}{i}.
\]  

Using the expressions of $J^{d}_x$,$J^{d}_y$ and $J^{d}_z$ the $\bar{J}^{d}$ can be written as
\begin{align}
    \bar{J}^{d}=J^{d}+\frac{2\alpha_3}{d-1}\sum_{i=0}^{d-1}\left(i - \frac{d-1}{2}\right)\ketbra{i}{i}.
    \label{off}
\end{align}
The $J^{d}$ is same as in Eq.~\ref{eq: Jd}.
Note that when \(\alpha_{3}=0\), \(\bar{J}^{d}\) reduces to the purely off-diagonal Hamiltonian \(J^{d}\), while for \(\alpha_{3} \neq 0\), \(\bar{J}^{d}\) contains both diagonal and off-diagonal elements in the coherence basis. 

 For $d=2$, from the analysis presented in Corollary 1, we have $\xi^p_{2}(\mathcal{C})=\xi_{2}(\mathcal{C})$. Also, $\xi^p_{2}(\mathcal{C})$ have the exact form as given in Eq.~\eqref{eq: xi2}, with $h_1$, $h_3$ being the diagonal elements of $\bar{J_2}$ and $h_2$ being the off-diagonal element of $\bar{J_2}$, when expressed in the coherence basis. For higher dimension we employ numerical optimization technique.
 
We begin with the numerical analysis of the qutrit case, using the same procedure outlined in Sec.~\ref{subsection: diagonal hamiltonian}, but with \(J^{d}_z\) replaced by \(\bar{J}_d\). Our numerical data confirm that for \(d=3\) and the Hamiltonian \(\bar{J}_d\), it is sufficient to consider only pure input states to achieve the CCMW. Thus even for this set up we have $\xi^p_{3}(\mathcal{C})=\xi_{3}(\mathcal{C})$

Encouraged by this finding, we extend our analysis to higher dimensions by restricting the optimization to pure initial states. Thus for dimensions \(d>3\), we numerically compute \(\xi^p_{d}(\mathcal{C})\) for various fixed input coherence values \(\mathcal{C}\). The steps followed for the numerical optimization is similar to that presented in Sec.~\ref{sec: off_diagonal_hamiltonian}, for higher dimension $d>3$, with the Hamiltonian $J^d$ now being replaced by $\bar{J}^d$.

In Fig.~\ref{fig: mixdiagonals} we plot \(\xi^p_{d}(\tilde{\mathcal{C}})\) against the scaled coherence \(\tilde{\mathcal{C}}\) for dimensions \(d = 2,3,4,5,6\) and for various values of \(\bar{\alpha} = \alpha/\alpha_3\). We set \(\alpha_3 = 1\) and consider \(\bar{\alpha} = 0.25, 0.5, 1, 5\). The figure shows that the behavior of \(\xi^p_{d}(\tilde{\mathcal{C}})\) depends on \(\bar{\alpha}\). For small \(\bar{\alpha}\) values (0.25 and 0.5), the two leftmost plots show that \(\xi^p_{d}(\tilde{\mathcal{C}})\) decreases as coherence increases. In contrast, for larger values (1 and 5), shown in the two rightmost plots, \(\xi^p_{d}(\tilde{\mathcal{C}})\) increases with coherence.

This difference arises because, with \(\alpha_3 = 1\), the parameter \(\bar{\alpha}\) controls the strength of the off‑diagonal elements in \(\bar{J}^{d}\) (see Eq.~\eqref{off}). Smaller \(\bar{\alpha}\) makes \(\bar{J}^{d}\) more similar to a diagonal Hamiltonian, mirroring the behavior described in Sec.~\ref{subsection: diagonal hamiltonian} where \(\xi^p_{d}(\tilde{\mathcal{C}})\)  decreases with coherence. Larger \(\bar{\alpha}\) imparts a more off-diagonal character, matching the trend in Sec.~\ref{sec: off_diagonal_hamiltonian} where \(\xi^p_{d}(\tilde{\mathcal{C}})\)  increases with coherence.

Further in this case one can easily compute the scaling of $\xi^p_{d}(\tilde{\mathcal{C}})$ with dimension $d$, for $\tilde{\mathcal{C}}=1$. Interestingly for $\tilde{\mathcal{C}}=1$, $\xi^p_{d}(\tilde{\mathcal{C}})$ scales with dimension as $$\xi^p_{d}(1)=4\alpha\left(1-{1}/{d}\right).$$

\section{Conclusion} \label{sec: conclusion}
In this work, we explored the relationship between an intrinsic quantum feature, namely quantum coherence, and the extractable work from quantum systems. Specifically, we considered practical scenarios where the objective is to extract the maximum amount of work from quantum systems that possess a limited or fixed level of coherence, with the additional requirement that the extraction process preserves this coherence. This constraint is significant because coherence serves as a key resource in various quantum communication and information-processing tasks. Therefore, protocols that conserve coherence during energy extraction are, first of all,  practically meaningful, and secondly,  instrumental in establishing a theoretical relationship between the constrained coherence input and the corresponding extractable energy. Such a relationship is, in turn, potentially crucial for the development of quantum technologies, especially quantum batteries that must operate under realistic experimental constraints.

We refer to the extractable work, which is maximized over all quantum states (pure or mixed) with a fixed level of coherence and over all corresponding unitaries that preserve this level of coherence, as the coherence‑constrained maximal work. We investigated how CCMW depends on a given value of the initial quantum coherence of the quantum battery with respect to an arbitrary but fixed basis. We began by analyzing qubit systems and derived an exact analytical relationship between CCMW and the fixed input coherence.
Our results revealed that CCMW is always achieved by pure states of the battery. Furthermore, we demonstrated that the relationship between CCMW and coherence exhibits a strong dependence on the choice of basis in which the coherence is considered. In particular, when the fixed coherence is considered in the eigenbasis of the Hamiltonian, CCMW decreases as the initial coherence increases. In contrast, when the Hamiltonian has equal or zero diagonal elements and non-zero off-diagonal elements in the coherence basis, CCMW increases with increasing input coherence. The basis-dependent behavior in work extraction was also observed in higher-dimensional systems.

We performed numerical analysis for higher dimensional batteries. Our numerical results shown that when coherence is fixed in the energy eigenbasis, restricting to pure initial states with fixed coherence suffices to reach the CCMW even in higher dimensions. Building on this numerical finding, we derived exact analytical expressions for the qutrit battery and provided a general framework for extending such results to even higher dimensions. Additionally, within this framework, we identified passive states in arbitrary dimensions that possess fixed coherence and from which no energy can be extracted when coherence must be conserved. Finally, for cases where the Hamiltonian is non-diagonal in the coherence basis, our qutrit-based numerical analysis guided the derivation of a closed-form relation linking CCMW to the fixed input coherence. Throughout this study, we used the $l_1$ norm to quantify coherence.

\bibliography{refs}

\appendix
\section{CCMW for Qutrit Batteries considering Hamiltonian with only Off-diagonal Elements }\label{appendix: non-diagonal qutrit}
In this appendix, we calculate the CCMW extracted by coherence-conserving unitaries acting on pure states for qutrit batteries, specifically when the system Hamiltonian contains only off-diagonal elements in the coherence basis \(\{\ket{i}\}\).

Consider the initial and final states of the qutrit battery as $\ketbra{\psi_{\text{in}}}{\psi_{\text{in}}} =\rho_{\text{in}}\in\zeta^{\mathcal{C}}_{3}$ and $\ketbra{\psi_{\text{f}}}{\psi_{\text{f}}}=\rho_{\text{f}}\in\zeta^{\mathcal{C}}_{3}$ respectively, such that in the given coherence basis $\{\ket{0},\ket{1},\ket{2}\}$,  
\begin{align*}
    &\ket{\psi_{\text{in}}}=x_{0}\ket{0}+x_{1}e^{i\theta_1}\ket{1}+x_{2}e^{i\theta_2}\ket{2}, \notag \\
    &\ket{\psi_{\text{f}}}=y_{0}\ket{0}+y_{1}e^{i\phi_1}\ket{1}+y_{2}e^{i\phi_2}\ket{2}.
\end{align*} Here $0\leq\{x_0,x_1,x_2,y_0,y_1,y_2\}\leq1$ and $0\leq\{\theta_1,\theta_2,\phi_1,\phi_2\}\leq2\pi$. Also, to satisfy the normalization condition and coherence preserving condition $\{x_0,x_1,x_2\}$ and $\{y_0,y_1,y_2\}$ must satisfy 
\begin{align}\label{eq: conditions1}
    x^2+y^2+z^2=1,\quad \text{and} \notag \quad xy+yz+xz=\frac{\mathcal{C}}{2} \\
\end{align}
or,
\begin{align}\label{eq: condition2}
    x^2+y^2+z^2=1,\quad \text{and} \quad x+y+z=\sqrt{1+\mathcal{C}}
\end{align}
as discussed in details in subsection~\ref{subsection: diagonal hamiltonian}.
The Hamiltonian of the qutrit battery is given in Eq.~\eqref{eq: Jd}, namely, 
\begin{equation}
    J^{3}=\alpha e^{I\phi}\left(\ketbra{0}{1}+\ketbra{1}{2}\right)+h.c.
\end{equation}
Then the unitarily extracted CCMW becomes 
\begin{equation*}
    \xi_{3}\left(\mathcal{C}\right)=\max_{\{\rho_{\text{in}},\rho_f\} \in \zeta^{\mathcal{C}}_3} \Tr[J^{3} (\rho_{\text{in}} - \rho_f)]
\end{equation*}

Note that we can again separate the optimizations over initial and final state parameters, as
\begin{align}
\xi_{3}\left(\mathcal{C}\right) &= \max_{\theta_1,\theta_2 x_0,x_1,x_2} 2\alpha \Big[ 
     x_0x_1\cos\left(\phi+\theta_1\right)\notag \\ & \qquad \qquad\qquad \quad+ x_1x_2\cos\left(\phi+\theta_2-\theta_1\right)\Big] \nonumber\notag \\
    & - \min_{\phi_1,\phi_2, y_0,y_1,y_2}2\alpha\Big[y_0y_1\cos\left(\phi+\phi_1\right)\notag  \\&\qquad \qquad\qquad \quad+y_1y_2\cos\left(\phi+\phi_2-\phi_1\right) 
\Big].
\end{align}
After independently optimizing over $\{\theta_1,\theta_2,\phi_1,\phi_2\}$, as they do not depend on other parameters, we get
\begin{align*}
    \xi_{3}\left(\mathcal{C}\right)=2\alpha\max_{x_0,x_1,x_2}(x_0x_1+x_1x_2)+2\alpha\max_{y_0,y_1,y_2}(y_0y_1+y_1y_2).
\end{align*}
Now using the second equation of Eq.~\eqref{eq: conditions1} and noticing that the $x_j$ and $y_{j}$ are dummy variables of optimizations and are constrained to satisfy same constraint equations, namely Eq.~\eqref{eq: conditions1} we can write
\begin{align}\label{eq: final_in_xyz}
    &\xi_{3}\left(\mathcal{C}\right)=4\alpha\max_{x,z}(\frac{\mathcal{C}}{2}-xz)\notag \\\implies \quad&\xi_{3}\left(\mathcal{C}\right)=4\alpha\left[\frac{\mathcal{C}}{2}-\min_{x,z}(xz)\right].
\end{align} Here $\{x,z\}$ satisfies
\begin{equation}\label{eq: xz constrain}
    zx=\left[(x+z)-\frac{\sqrt{1+\mathcal{C}}}{2}\right]^2+\frac{\mathcal{C}-1}{4},
\end{equation}
which one can get by substituting $y$ from second part of Eq.~\eqref{eq: condition2} to the first part of Eq.~\eqref{eq: condition2}. 

Now we perform a change of variable: 
\begin{equation}\label{eq: pq_transformation}
    xz=q, \quad (x+y)=p.
\end{equation}
After this transformation, the equation Eq.~\eqref{eq: final_in_xyz} becomes 
\begin{equation}
    \xi_{3}\left(\mathcal{C}\right)=4\alpha\left[\frac{\mathcal{C}}{2}-\min_{q}(q)\right],
\end{equation} and the Eq.~\ref{eq: xz constrain} becomes 
\begin{equation}\label{eq: pq_constraint}
    q=\left(p-\frac{\sqrt{1+\mathcal{C}}}{2}\right)^2+\frac{\mathcal{C}-1}{4}.
\end{equation}

This equation describes a parabola in the \(p\)-\(q\) plane for a given coherence \(\mathcal{C}\). Our goal is to find the minimum \(q\) along this parabola. Note that \(p\) and \(q\) must also satisfy the constraints
\[
p^2 \ge 4q \quad\text{and}\quad q \ge 0,
\]
otherwise no real, positive solutions for \(x\) and \(z\) exist to satisfy the transformation in Eq.~\eqref{eq: pq_transformation}. We refer to the region defined by these inequalities as the valid solution region.

Thus, the task of calculating \(\xi_3(\mathcal{C})\) in Eq.~\eqref{eq: final_in_xyz} reduces to finding the minimum value of \(q\) that satisfies the parabola's equation in Eq.~\eqref{eq: pq_constraint} within this valid region.

For \(\mathcal{C} \le 1\), the parabola attains its minimum at a negative \(q\), which lies outside the valid region. Therefore, the minimum \(q\) in this case is zero.  
For \(1 \le \mathcal{C} \le {5}/{3}\), the parabola's minimum occurs at \(q = (\mathcal{C} - 1) / 4\), which lies within the valid region.  
For \(\mathcal{C} \ge {5}/{3}\), the parabola's minimum again falls outside the valid region. In this case, the minimum \(q\) occurs where the parabola intersects the boundary \(p^2 = 4q\), yielding  
\[
q = \frac{\left(\sqrt{1+\mathcal{C}} - \sqrt{1 - \tfrac{\mathcal{C}}{2}}\right)^2}{9}.
\]

Combining these three regimes, we obtain
\[
\xi_3(\mathcal{C}) = 
\begin{cases}
2\alpha\,\mathcal{C}, & 0 \le \mathcal{C} \le 1, \\
\alpha\,(\mathcal{C} + 1), & 1 \le \mathcal{C} \le \tfrac{5}{3}, \\
2\alpha\left(\mathcal{C} - f(\mathcal{C})\right), & \tfrac{5}{3} \le \mathcal{C} \le 2,
\end{cases}
\]
where
\[
f(\mathcal{C}) = \frac{2}{9}\left(\sqrt{1 + \mathcal{C}} - \sqrt{1 - \tfrac{\mathcal{C}}{2}}\right)^2.
\]
Thus we have the exact expression for CCMW, for the qutrit case when the Hamiltonian poses only off-diagonal elements in the coherence basis.
\end{document}